\documentclass[11 pt]{article}

\def\showauthornotes{1}

\def\showdraftbox{0}

\usepackage{xspace,xcolor,enumerate}
\usepackage{amsmath,amssymb}
\usepackage{amsthm}
\usepackage[toc,page]{appendix}
\usepackage{thmtools}
\usepackage{thm-restate}
\usepackage{color,graphicx}
\usepackage{authblk}

\definecolor{darkred}{rgb}{0.5,0,0}
\definecolor{darkgreen}{rgb}{0,0.35,0}
\definecolor{darkblue}{rgb}{0,0,0.55}

\usepackage[pdfstartview=FitH,pdfpagemode=UseNone,colorlinks,linkcolor=darkblue,filecolor=darkred,citecolor=darkgreen,urlcolor=darkred,pagebackref]{hyperref}

\usepackage[capitalise,nameinlink]{cleveref}
\usepackage[T1]{fontenc}
\usepackage{mathtools,dsfont,bbm}

 \usepackage{mathpazo}
\usepackage[top=1in, bottom=1in, left=1in, right=1in]{geometry}

\allowdisplaybreaks

\ifnum\showauthornotes=1
\newcommand{\Authornote}[2]{{\sf\small\color{red}{[#1: #2]}}}
\newcommand{\Authorcomment}[2]{{\sf \small\color{gray}{[#1: #2]}}}
\newcommand{\Authorfnote}[2]{\footnote{\color{red}{#1: #2}}}
\else
\newcommand{\Authornote}[2]{}
\newcommand{\Authorcomment}[2]{}
\newcommand{\Authorfnote}[2]{}
\fi

\ifnum\showdraftbox=1

\else

\fi

\newtheorem{theorem}{Theorem}[section]

\newtheorem{observation}[theorem]{Observation}
\newtheorem{definition}[theorem]{Definition}
\newtheorem{lemma}[theorem]{Lemma}
\newtheorem{remark}[theorem]{Remark}

\newtheorem{corollary}[theorem]{Corollary}
\newtheorem{claim}[theorem]{Claim}
\newtheorem{fact}[theorem]{Fact}

\def\FullBox{\hbox{\vrule width 6pt height 6pt depth 0pt}}

\def\qed{\ifmmode\qquad\FullBox\else{\unskip\nobreak\hfil
\penalty50\hskip1em\null\nobreak\hfil\FullBox
\parfillskip=0pt\finalhyphendemerits=0\endgraf}\fi}

\def\qedsketch{\ifmmode\Box\else{\unskip\nobreak\hfil
\penalty50\hskip1em\null\nobreak\hfil$\Box$
\parfillskip=0pt\finalhyphendemerits=0\endgraf}\fi}

\def\to{\rightarrow}
\def\eps{\varepsilon}
\def\epsilon{\varepsilon}

\def\eps{\epsilon}

\def\phi{\varphi}
\def\cal{\mathcal}

\newcommand{\defeq}{:=}

\newcommand{\ie}{i.e.,\xspace}

\newcommand{\etal}{et al.\xspace}

\newcommand{\mper}{\,.}
\newcommand{\mcom}{\,,}

\newcommand{\R}{{\mathbb R}}

\newcommand{\C}{{\mathbb C}}
\newcommand{\N}{{\mathbb{N}}}

\newcommand{\pmone}{\{-1,1\}\xspace}

\newcommand{\gauss}[2]{{\cal N(#1, #2)}}
\newcommand{\gaussian}[2]{{\cal N(#1, #2)}}

\usepackage{nicefrac}

\let\nfrac=\nicefrac

\newcommand{\abs}[1]{\ensuremath{\left\lvert #1 \right\rvert}}

\makeatletter
\def\norm#1{
  \@ifnextchar\bgroup%
   {\normalpnorm{#1}}%
   {\defaultnorm{#1}}%
}
\def\defaultnorm#1{%
    \renderNorm{#1}{}
}
\def\normalpnorm#1#2{%
   \@ifnextchar\bgroup%
   {\banaxnorm{#1}{#2}}%
   {\renderNorm{#2}{#1}}
}

\def\banaxnorm#1#2#3{%
    \renderNorm{#3}{#1\to#2 }
}

\def\renderNorm#1#2{%
    \@ifnextchar^%
    {\fixExponent{#1}{#2}}%
    {\ensuremath{\mathchoice%
        {\lVert #1 \rVert_{#2}}%
        {\lVert #1 \rVert_{#2}}%
        {\lVert #1 \rVert_{#2}}%
        {\lVert #1 \rVert_{#2}}}%
    }%
}
\def\fixExponent#1#2^#3{%
    \ensuremath{{\lVert #1 \rVert^{#3}_{#2}}}%
}

\def\enorm#1#2{%
   \@ifnextchar\bgroup%
   {\norm{L_{#1}}{L_{#2}}}%
   {\norm{L_{#1}}{#2}}%
}

\def\cnorm#1#2{%
   \@ifnextchar\bgroup%
   {\norm{\ell_{#1}}{\ell_{#2}}}%
   {\norm{\ell_{#1}}{#2}}%
}

\makeatother

\newcommand{\mydot}[2]{\ensuremath{\left\langle #1\,, #2 \right\rangle}}
\newcommand{\mysmalldot}[2]{\ensuremath{\langle #1\,, #2 \rangle}}
\newcommand{\ip}[1]{\left\langle #1 \right\rangle}

\def\bfg{{\bf g}}

\newcommand{\Esymb}{\mathbb{E}}

\makeatletter

\def\Ex#1{%
    \ProbabilityRender{\Esymb}{#1}%
}

\def\ProbabilityRender#1#2{
  \@ifnextchar\bgroup%
  {\renderwithdist{#1}{#2}}
   {\singlervrender{#1}{#2}}
}
\def\singlervrender#1#2{%
   \ensuremath{\mathchoice
       {{#1}\left[ #2 \right]}
       {{#1}[ #2 ]}
       {{#1}[ #2 ]}
       {{#1}[ #2 ]}
   }
}
\def\renderwithdist#1#2#3{%
   \@ifnextchar\bgroup
   {\superfancyrender{#1}{#2}{#3}}
   {\ensuremath{\mathchoice
      {\underset{#2}{#1}\left[ #3 \right]}
      {{#1}_{#2}[ #3 ]}
      {{#1}_{#2}[ #3 ]}
      {{#1}_{#2}[ #3 ]}
     }
   }
}
\def\superfancyrender#1#2#3#4#5{
   \ensuremath{\mathchoice
      {\underset{#1}{{#1}}\left#4 #3 \right#5}
      {{#1}_{#2}#4 #3 #5}
      {{#1}_{#2}#4 #3 #5}
      {{#1}_{#2}#4 #3 #5}
   }
}
\makeatother

\newfont{\inhead}{eufm10 scaled\magstep1}
\newcommand{\deffont}{\sf}
\newcommand{\defn}[1]{{\deffont #1}}

\newcommand{\calF}{{\cal F}}

\DeclareMathOperator{\sgn}{\operatorname{sgn}}

\DeclareMathOperator*{\argmax}{\arg\!\max}

\DeclareMathOperator{\diag}{\operatorname{diag}}

\newcommand{\ee}{\ensuremath{\mathrm e}}

\newcommand{\1}[1]{\mathds{1}\left[#1\right]}

\newcommand{\classfont}[1]{\textsf{#1}}

\newcommand{\NP}{\classfont{NP}\xspace}

\newcommand{\onorm}[2]{\ensuremath{#1{\to}#2} norm}

\newcommand{\inparen}[1]{\left(#1\right)}             
\newcommand{\inbraces}[1]{\left\{#1\right\}}           
\newcommand{\insquare}[1]{\left[#1\right]}             

\newcommand{\Holder}{H\"{o}lder\xspace}

\newcommand{\Odonnell}{O'Donnell\xspace}
\newcommand{\Briet}{Bri\"et\xspace}          
\newcommand{\Kwapien}{Kwapie\'n\xspace}

\newcommand{\fplain}[3]{\widetilde{f}_{#1,\, #2}(#3)}

\newcommand{\fin}[3]{\widetilde{f}^{-1}_{#1,\, #2}(#3)}
\newcommand{\alfin}[2]{\widetilde{f}^{-1}_{#1,\, #2}}

\newcommand{\alfinabs}[2]{\widetilde{h}_{#1,\, #2}}
\newcommand{\hin}[3]{\widetilde{h}^{-1}_{#1,\, #2}(#3)}

\newcommand{\nf}[3]{f_{#1,\, #2}(#3)}
\newcommand{\alnf}[2]{f_{#1,\, #2}}
\newcommand{\nfext}[3]{f^{+}_{#1,\, #2}(#3)}

\newcommand{\nfin}[3]{f^{-1}_{#1,\, #2}(#3)}
\newcommand{\alnfin}[2]{f^{-1}_{#1,\, #2}}
\newcommand{\nh}[3]{h_{#1,\, #2}(#3)}
\newcommand{\alnh}[2]{h_{#1,\, #2}}
\newcommand{\nhin}[3]{h^{-1}_{#1,\, #2}(#3)}

\newcommand{\absolute}[1]{\mathrm{abs}\inparen{#1}}

\newcommand{\kfinc}[1]{\widetilde{f}^{-1}_{#1}}
\newcommand{\ngc}{f^{\,-1}_{k}}

\newcommand{\fa}[1]{\widetilde{f}^{\,(a)}(#1)}
\newcommand{\fb}[1]{\widetilde{f}^{\,(b)}(#1)}
\newcommand{\fwild}[2]{\widetilde{f}^{\,(#1)}(#2)}
\newcommand{\alfa}{\widetilde{f}^{\,(a)}}
\newcommand{\alfb}{\widetilde{f}^{\,(b)}}
\newcommand{\alfwild}[1]{\widetilde{f}^{\,(#1)}}
\newcommand{\hfac}[1]{\widehat{f}^{\,\,(a)}_{#1}}
\newcommand{\hfbc}[1]{\widehat{f}^{\,\,(b)}_{#1}}
\newcommand{\hfwildc}[2]{\widehat{f}^{\,\,(#1)}_{#2}}

\newcommand{\intfull}[1]{\mathrm{I}(#1)}

\newcommand{\betaterm}{\mathrm{B}\inparen{\frac{1-b}{2},1+\frac{b}{2}}}

\newcommand{\confHG}{{}_{1}F_{1}}
\newcommand{\hypergeometric}{{}_{2}F_{1}}
\newcommand{\hgp}[1]{#1\cdot \hypergeometric\!\inparen{\frac{1-a}{2},\frac{1-b}{2}\,;\,\frac{3}{2}\,;\,{#1}^2}}
\newcommand{\RisingFactorial}[1]{(#1)_{k}}

\newcommand{\tildeU}{\widetilde{U}}
\newcommand{\tildeV}{\widetilde{V}}

\newcommand{\barx}{\overline{x}}
\newcommand{\bary}{\overline{y}}

\newcommand{\tildex}{\widetilde{x}}
\newcommand{\tildey}{\widetilde{y}}

\newcommand{\hatU}{\widehat{U}}
\newcommand{\hatV}{\widehat{V}}

\newcommand{\seriesLeft}[2]{S_L(#1,#2)}
\newcommand{\seriesRight}[2]{S_R(#1,#2)}

\newcommand{\gvec}{\bfg}

\newcommand{\bfX}{\mathbf{X}}

\newcommand{\holderdual}[2]{\Psi_{\!#1}(#2)}

\newcommand{\id}{\mathrm{I}}

\newcommand{\CP}[1]{\mathsf{CP}(#1)}
\newcommand{\CPD}[1]{\mathrm{DP}(#1)}

\newcommand{\HD}{\emph{\Holder Dual Rounding}~}

\newcommand{\Sym}{\mathbb{S}}

\newcommand{\noise}[1]{T_{\rho} \,{#1}}

\newcommand{\factorConst}[1]{\Phi(#1)}
\newcommand{\threeFactorConst}[1]{\Phi_{3}(#1)}

\newcommand{\factorSpConst}[2]{\Phi(#1,#2)}
\newcommand{\threeFactorSpConst}[2]{\Phi_{3}(#1,#2)}

\newcommand{\sqnormX}{\mathcal{F}_{X}}
\newcommand{\sqnormY}{\mathcal{F}_{Y}}
\newcommand{\sqrtnormX}{\sqrt{\mathcal{F}}_{X}}
\newcommand{\sqrtnormY}{\sqrt{\mathcal{F}}_{Y}}

\newcommand{\bbX}{\mathbb{X}}
\newcommand{\bbY}{\mathbb{Y}}
\newcommand{\bbW}{\mathbb{W}}

\newcommand{\bars}{\overline{s}}
\newcommand{\bart}{\overline{t}}

\newcommand{\Diag}[1]{D_{#1}}

\newcommand{\Ball}[1]{\mathrm{Ball}(#1)}

\newcommand{\qedforce}{\tag*{$\blacksquare$}}

\begin{document}

\title{Approximating Operator Norms via Generalized Krivine Rounding
\footnote{An extended abstract of this work (without proofs and details) has been published in the conference proceedings of 
the 2019~~Symposium on Discrete Algorithms}
}

\author[1]{
~~~~Vijay Bhattiprolu\thanks{Supported by NSF CCF-1422045 and CCF-1526092.  {\tt
vpb@cs.cmu.edu}. Part of the work was done while visiting UC Berkeley and CMSA, Harvard.} 
}
\author[2]{ 
~~Mrinalkanti Ghosh\thanks{Supported by NSF CCF-1254044 {\tt mkghosh@ttic.edu}} 
}
\author[1]{
~~Venkatesan Guruswami\thanks{Supported in part by NSF grant CCF-1526092. {\tt guruswami@cmu.edu}. 
Part of the work was done while visiting CMSA, Harvard. }
}

\author[3]{
$\qquad$ Euiwoong Lee\thanks{Supported by the Simons Institute for the Theory of Computing. 
{\tt euiwoong@cims.nyu.edu}} 
}
\author[2]{
~Madhur Tulsiani \thanks{Supported by NSF CCF-1254044 {\tt madhurt@ttic.edu}} 
}

\affil[1]{Carnegie Mellon University}
\affil[2]{Toyota Technological Institute Chicago}
\affil[3]{New York University}

\setcounter{page}{0}
\date{}

\maketitle

\begin{abstract}

We consider the $(\ell_p,\ell_r)$-Grothendieck problem, which seeks to
maximize the bilinear form $y^T A x$ for an input matrix $A \in
{\mathbb R}^{m \times n}$ over vectors $x,y$ with
$\|x\|_p=\|y\|_r=1$. The problem is equivalent to computing the $p \to
r^\ast$ operator norm of $A$, where $\ell_{r^*}$ is the dual norm to
$\ell_r$. The case $p=r=\infty$ corresponds to the classical
Grothendieck problem. Our main result is an algorithm
for arbitrary $p,r \ge 2$ with approximation ratio 
$(1+\epsilon_0)/(\sinh^{-1}(1)\cdot \gamma_{p^\ast} \,\gamma_{r^\ast})$ for
some fixed $\epsilon_0 \le 0.00863$. Here $\gamma_t$ denotes the $t$'th
norm of the standard Gaussian. Comparing this with Krivine's approximation ratio 
$(\pi/2)/\sinh^{-1}(1)$ for the original Grothendieck problem, our guarantee is off from the best known
hardness factor of $(\gamma_{p^\ast} \gamma_{r^\ast})^{-1}$ for the
problem by a factor similar to Krivine's defect (up to the constant $(1+\epsilon_0)$).
\smallskip

Our approximation follows by bounding the value of the
natural vector relaxation for the problem which is convex when $p,r
\ge 2$. We give a generalization of random hyperplane rounding using
H\"{o}lder-duals of Gaussian projections rather than taking the
sign. We relate the performance of this rounding to certain
hypergeometric functions, which prescribe necessary transformations to
the vector solution before the rounding is applied. Unlike
Krivine's Rounding where the relevant hypergeometric function was
$\arcsin$, we have to study a family of hypergeometric
functions. The bulk of our technical work then involves methods from
complex analysis to gain detailed information about the Taylor series
coefficients of the inverses of these hypergeometric functions, which
then dictate our approximation factor. 
\smallskip

Our result also implies improved bounds for ``factorization through $\ell_{2}^{n}$'' of operators 
from $\ell_{p}^{n}$ to $\ell_{q}^{m}$ (when $p\geq 2 \geq q$), and our work
provides modest supplementary evidence for an intriguing parallel between factorizability, and
constant-factor approximability.

\end{abstract}

\thispagestyle{empty}
\newpage

\tableofcontents
\thispagestyle{empty}
\newpage

\pagenumbering{arabic}

\setcounter{page}{1}

\section{Introduction}
We consider the problem of finding the \onorm{p}{q} of a given matrix $A \in \R^{m \times n}$,
which is defined as 
\[
\norm{p}{q}{A} ~\defeq~ \max_{x \in \R^n \setminus \{0\}} \frac{\norm{q}{Ax}}{\norm{p}{x}} \mper 
\]
The quantity $\norm{p}{q}{A}$ is a natural generalization of the well-studied spectral norm
($p=q=2$) and computes the maximum distortion (stretch) of the operator $A$ from the normed space
$\ell_p^n$ to $\ell_q^m$.
The case when $p = \infty$ and $q = 1$ is the well known Grothendieck problem \cite{KN12, Pisier12},
where the goal is to maximize $\ip{y, Ax}$ subject to $\norm{\infty}{y}, \norm{\infty}{x} \leq 1$. In
fact, via simple duality arguments, the general problem computing $\norm{p}{q}{A}$ can be seen to be
equivalent to the following variant of the Grothendieck problem
\[
    \norm{p}{q}{A}
    ~=~
    \max_{\substack{\norm{p}{x} \leq 1\\\norm{q^*}{y} \leq 1}}\ip{y, Ax}
    ~=~
    \norm{q^*}{p^*}{A^T} 
    \mcom
\]
where $p^*, q^*$ denote the dual norms of $p$ and $q$, satisfying $1/p + 1/p^* = 
1/q + 1/q^* = 1$. The above quantity is also known as the injective tensor norm of $A$ where $A$ is
interpreted as an element of  the space $\ell_{q}^{m}\otimes \ell_{p^*}^{n}$.

In this work, we consider the case of $p \geq q$, where the problem is known to admit good 
approximations when $2 \in [q,p]$, and is hard otherwise. Determining the right constants in these
approximations when $2 \in [q,p]$ has been of considerable interest in the analysis and optimization
community.

For the case of \onorm{\infty}{1}, Grothendieck's theorem \cite{Grothendieck56} shows that the
integrality gap of a semidefinite programming (SDP) relaxation is bounded by a constant, and the
(unknown) optimal value is now called the Grothendieck constant $K_G$. Krivine \cite{Krivine77} proved
an upper bound of $\pi/(2\ln(1+\sqrt{2})) = 1.782 \ldots$ on $K_G$, and it was later shown by
Braverman \etal \cite{BMMN13} that $K_G$ is strictly smaller than this bound. The best known lower bound 
on $K_G$ is about $1.676$, due to (an unpublished manuscript of) Reeds \cite{Reeds91} (see also
\cite{KO09} for a proof).

A very relevant work of Nestereov \cite{Nesterov98} proves an upper bound of $K_G$ on the
approximation factor for \onorm{p}{q} for any $p \geq 2 \geq q$ (although the bound stated there is
slightly weaker - see \cref{KG:worst:case} for a short proof).
A later work of Steinberg \cite{Steinberg05} also gave an upper bound of 
$\min\inbraces{\gamma_p/\gamma_q, \gamma_{q^*}/\gamma_{p^*}}$, where $\gamma_p$ denotes
$p^{th}$  norm of a standard normal random variable (\ie the $p$-th root of the $p$-th Gaussian moment). 

On the hardness side, \Briet, Regev and Saket \cite{BRS15} showed NP-hardness of $\pi/2$ for the
\onorm{\infty}{1} (in fact it even holds for the PSD-Grothendieck problem),  strengthening a hardness 
result of Khot and Naor based on the Unique Games Conjecture (UGC) \cite{KN09} (which also improves on 
the previously known NP-Hardness due to \cite{AN04} via Max-Cut). Assuming UGC, a hardness result matching Reeds' lower bound was proved by Khot and
\Odonnell \cite{KO09}, and hardness of approximating within $K_G$ was proved by Raghavendra and
Steurer \cite{RS09}. In preceding work \cite{BGGLT18a}, the authors proved NP-hardness of
approximating \onorm{p}{q} within any factor better than $1/(\gamma_{p^*} \cdot \gamma_q)$, for any
$p \geq 2 \geq q$. Stronger hardness results are known and in particular the problem admits no
constant approximation, for the cases not considered in this paper \ie when $p \leq q$ or $2 \notin
[q,p]$. We refer the interested reader to a detailed discussion in \cite{BGGLT18a}. 

\subsection{The Search For Optimal Constants and Optimal Algorithms}
The goal of determining the right approximation ratio for these problems is closely related to the
question of finding the optimal rounding algorithms (\ie algorithms that map the output of a convex 
programming relaxation to a feasible solution of the original optimization problem). 
The well known Hyperplane rounding procedure is widely applicable in combinatorial optimization 
because such problems can mostly be cast as optimization problems over the hypercube, and the 
output of hyperplane rounding is always a vertex of the hypercube. In a similar manner, many 
rounding algorithms tend to find usage in several optimization problems. It is thus a natural goal to 
develop such rounding algorithms (and the tools to analyze them) when the feasible domain 
is a more general convex set and in our case, the $\ell_p$ unit ball. \medskip 

For the Grothendieck problem, the goal is to
find $y \in \R^m$ and $x \in \R^n$ with $\norm{\infty}{y}, \norm{\infty}{x} \leq 1$, and one considers
the following semidefinite relaxation:
\begin{align*}
    \mbox{maximize} \quad&~~\sum_{i,j} A_{i,j}\cdot \mysmalldot{u^i}{v^j} \quad  \text{s.t.} \\
\mbox{subject to}    \quad&~~\norm{2}{u^i} \leq 1, \norm{2}{v^j} \leq 1 & \forall i\in [m], j\in [n]\\
    &~~ u^i, v^j\in \R^{m+n} & \forall i\in [m], j\in [n]
\end{align*}
By the bilinear nature of the problem above, it is clear that the optimal $x, y$ can be taken to
have entries in $\pmone$. A bound on the approximation ratio\footnote{Since we will be dealing with problems where the optimal solution may not be integral, we
  will use the term ``approximation ratio'' instead of ``integrality gap''.}
 of the above program is then obtained by
designing a good ``rounding'' algorithm which maps the vectors $u^i, v^j$ to values in
$\pmone$. Krivine's analysis \cite{Krivine77} corresponds to a rounding algorithm which considers a
random vector $\gvec \sim \gaussian{0}{I_{m+n}}$ and rounds to $x, y$ defined as
\[
y_i ~:=~ \sgn\inparen{\ip{\phi(u^i), \gvec}} \qquad \text{and} \qquad 
x_j 
~:=~
\sgn\inparen{\ip{\psi(v^j), \gvec}} \mcom
\]
for some appropriately chosen transformations $\phi$ and $\psi$. This gives the following 
upper bound on the approximation ratio of the above relaxation, and hence on the value of the
Grothendieck constant $K_G$:
\[
K_G 
~\leq~ 
\frac{1}{\sinh^{-1}(1)} \cdot \frac{\pi}{2} 
~=~
\frac{1}{\ln(1+\sqrt{2})} \cdot \frac{\pi}{2} \mper
\]
Braverman \etal \cite{BMMN13} show that the above bound can be strictly improved (by a very small
amount) using a two dimensional analogue of the above algorithm, where the value $y_i$ is taken to be
a function of the two dimensional projection $(\langle \phi(u^i), \gvec_1 \rangle, \langle \phi(u^i), \gvec_2
\rangle)$ for independent Gaussian vectors $\gvec_1, \gvec_2 \in \R^{m+n}$ (and similarly for
$x$). Naor and Regev \cite{NR14} show that such schemes are optimal in the sense that it is possible
to achieve an approximation ratio arbitrarily close to the true (but unknown) value of $K_G$ by
using $k$-dimensional projections for a large (constant) $k$. A similar existential result was also
proved by Raghavendra and Steurer \cite{RS09} who proved that the there exists a (slightly
different) rounding algorithm which can achieve the (unknown) approximation ratio $K_G$.

For the case of arbitrary $p \geq 2 \geq q$, Nesterov \cite{Nesterov98} considered the 
convex program in \cref{fig:convex}, denoted as $\CP{A}$, generalizing the one above.
\begin{figure}[ht]
\hrule
\vline
\begin{minipage}[t]{0.99\linewidth}
\vspace{-5 pt}
{\small
\begin{align*}
    \mbox{maximize} \quad&~~\sum_{i,j} A_{i,j}\cdot \mysmalldot{u^i}{v^j} ~~=~~  \ip{A, U V^T}\\
    \mbox{subject to}    \quad&~~\sum_{i\in [m]} \norm{2}{u^i}^{q^*} ~\leq~ 1 & \\
    &~~\sum_{j\in [n]} \norm{2}{v^j}^{p} ~\leq~ 1 & \\
    &~~ u^i, v^j\in \R^{m+n} & \forall i\in [m], j\in [n] \\
    u^i \text{ (resp. $v^j$) is the} & \text{ $i$-th (resp. $j$-th) row of $U$ (resp. $V$)} 
\end{align*}
}
\vspace{-14 pt}
\end{minipage}
\hfill\vline
\hrule
\caption{The relaxation $\CP{A}$ for approximating  \onorm{p}{q} of a matrix $A \in \R^{m \times n}$.} 
\label{fig:convex}
\end{figure}
Note that since $q^* \geq 2$ and $p \geq 2$, the above program is convex in the entries of the Gram
matrix of the vectors $\inbraces{u^i}_{i \in [m]}\cup  \inbraces{v^j}_{j \in [n]}$. Although the stated
bound in \cite{Nesterov98} is slightly weaker (as it is proved for a larger class of problems), the
approximation ratio of the above relaxation can be shown to be bounded by $K_G$. By using the
Krivine rounding scheme of considering the sign of a random Gaussian projection (aka random
hyperplane rounding) one can show that Krivine's upper bound on $K_G$ still applies to the above
problem. 

Motivated by applications to robust optimization, Steinberg \cite{Steinberg05} considered the dual
of (a variant of) the above relaxation, and obtained an upper bound of 
$\min\inbraces{\gamma_p/\gamma_q, \gamma_{q^*}/\gamma_{p^*}}$
on the approximation factor.
Note that while Steinberg's bound is better (approaches 1) 
as $p$ and $q$ approach 2, it is unbounded when $p, q^* \to \infty$ (as in the Grothendieck
problem).

Based on the inapproximability result of factor $1/(\gamma_{p^*} \cdot \gamma_q)$ obtained in 
preceding work by the authors \cite{BGGLT18a}, it is 
natural to ask if this is the ``right form'' of the approximation ratio. Indeed, this ratio is $\pi/2$
when $p^* = q = 1$, which is the ratio obtained by Krivine's rounding scheme, up to a factor of
$\ln(1+\sqrt{2})$. We extend Krivine's result to all $p \geq 2 \geq q$ as below.
\begin{theorem}
There exists a fixed constant $\eps_0 \leq 0.00863$ such that for all $p \geq 2 \geq q$, the
approximation ratio of the convex relaxation $\CP{A}$ is upper bounded by
\[
    \frac{1+\eps_0}{\sinh^{-1}(1)} \cdot \frac{1}{\gamma_{p^*} \cdot \gamma_q} 
~=~
    \frac{1+\eps_0}{\ln(1+\sqrt{2})} \cdot \frac{1}{\gamma_{p^*} \cdot \gamma_q}
\mper
\]
\end{theorem}
\begin{figure}
\begin{center}
\includegraphics[width=3.8in]{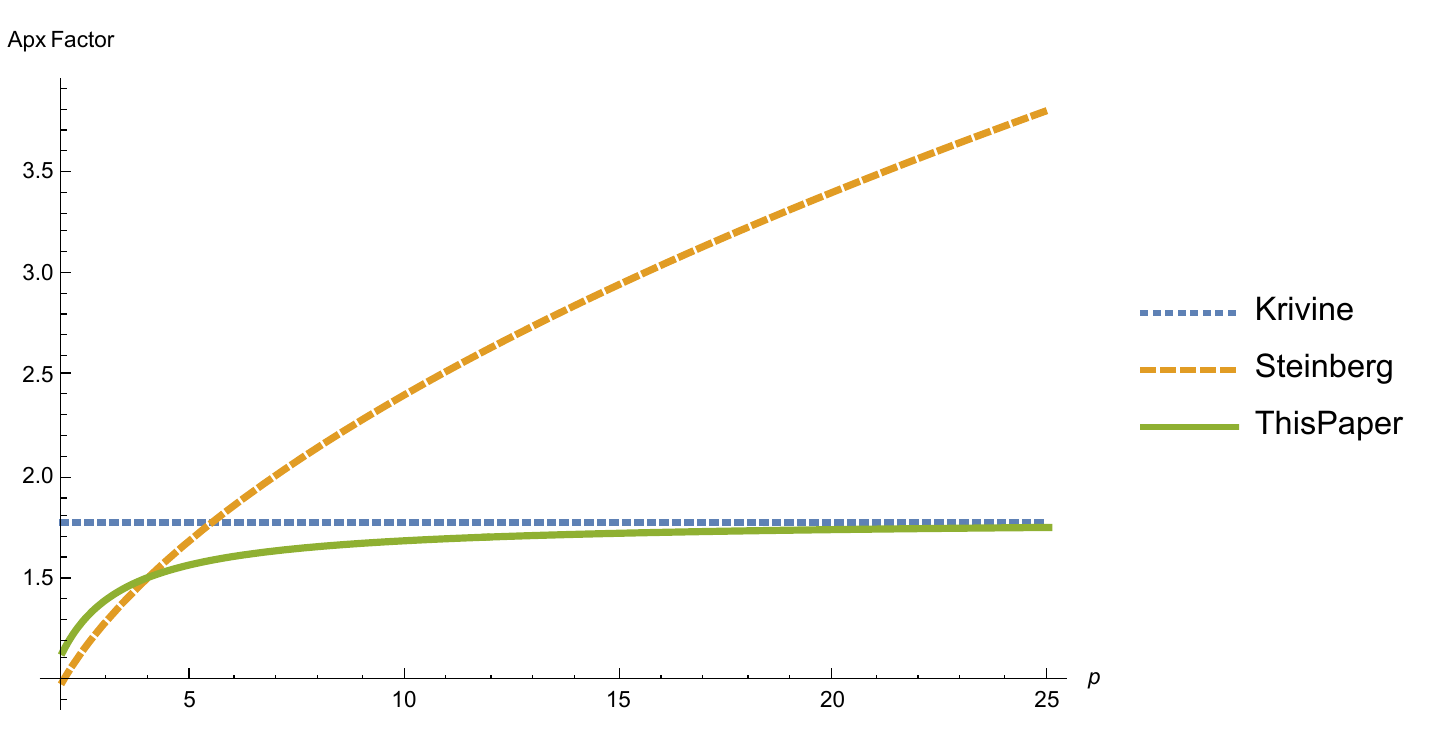}
\end{center}
\vspace{-15 pt}
\caption{{\footnotesize A comparison of the bounds for approximating \onorm{p}{p^*} obtained from Krivine's rounding
for $K_G$, Steinberg's analysis, and our bound. While our analysis yields an improved bound for $4\leq p\leq 66$, we believe that the rounding algorithm achieves an improved bound for all $p$.}}
\label{fig:bounds-comparison}
\end{figure}
Perhaps more interestingly, the above theorem is proved via a generalization of hyperplane rounding,
which we believe may be of independent interest. Indeed, for a given collection of vectors $w^{1},
\ldots, w^{m}$ considered as rows of a matrix $W$,  Gaussian hyperplane rounding corresponds to
taking the ``rounded'' solution $y$ to be the 
\[
    y 
    ~:=~ 
    \argmax_{\norm{\infty}{y'} \leq 1} \ip{y', W\gvec} 
    ~=~ 
    \inparen{\sgn\inparen{\ip{w^i,\gvec}}}_{i \in [m]}\mper
\]
We consider the natural generalization to (say) $\ell_r$ norms, given by
\[
    y 
    ~:=~ \argmax_{\norm{r}{y'} \leq 1} \ip{y', W\gvec} 
    ~=~ \inparen{\frac{\sgn\inparen{\ip{w^i, \gvec}} \cdot 
    \abs{\langle w^i, \gvec \rangle}^{r^*-1}}{\norm{r^*}{W\gvec}^{r^*-1}}}_{i \in [m]}\mper
\]
We refer to $y$ as the ``\Holder dual'' of $W\gvec$, since the above rounding can be obtained by
viewing $W\gvec$ as lying in the dual ($\ell_{r^*}$) ball, and finding the $y$ for which \Holder's
inequality is tight. Indeed, in the above language, Nesterov's rounding corresponds to considering
the $\ell_{\infty}$ ball (hyperplane rounding). While Steinberg used a somewhat different
relaxation, the rounding there can be obtained by viewing $W\gvec$ as lying in the primal $(\ell_r)$
ball instead of the dual one.
In case of hyperplane rounding, the analysis is motivated by the identity that for two unit vectors
$u$ and $v$, we have
\[
\Ex{\gvec}{\sgn(\ip{\gvec,u}) \cdot \sgn(\ip{\gvec,v})} ~=~ \frac{2}{\pi} \cdot \sin^{-1}(\ip{u,v}) \mper
\]
We prove the appropriate extension of this identity to $\ell_r$ balls (and analyze the functions
arising there) which may also be of interest for other optimization problems over $\ell_r$ balls.

\subsection{Proof overview}
As discussed above, we consider Nesterov's convex relaxation and generalize the hyperplane rounding
scheme using ``\Holder duals'' of the Gaussian projections, instead of taking the sign. 
As in the Krivine rounding scheme, this rounding is applied to transformations  of the SDP
solutions. The nature of these transformations depends on how the rounding procedure changes the
correlation between two vectors. Let $u,v \in \R^{N}$ be two unit vectors with $\ip{u,v} =
\rho$. Then, for $\gvec \sim \gauss{0}{I_{N}}$, $\ip{\gvec, u}$ and $\ip{\gvec, v}$ are
$\rho$-correlated Gaussian random variables. Hyperplane rounding then gives $\pm 1$ valued random
variables whose correlation is given by
\[
\Ex{\bfg_1 \sim_\rho \,\bfg_2}{\sgn(\bfg_1) \cdot \sgn(\bfg_2)}
    ~=~ 
\frac{2}{\pi} \cdot \sin^{-1}(\rho) \mper
\]
The transformations $\phi$ and $\psi$ (to be applied to the vectors $u$ and $v$) in Krivine's scheme
are then chosen depending on the Taylor series for the $\sin$ function, which is the inverse of
function computed on the correlation. For the case of \Holder-dual rounding, we prove the following
generalization of the above identity
\[
\Ex{\bfg_1 \sim_\rho \,\bfg_2}{\sgn(\bfg_1) \abs{\gvec_1}^{q-1} \cdot \sgn(\bfg_2)
  \abs{\gvec_2}^{p^*-1}}
~=~
\gamma_{q}^{q} \cdot  \gamma_{p^*}^{p^*} \cdot 
\rho \cdot \hypergeometric\!\inparen{1-\frac{q}{2},1-\frac{p^*}{2}\,;\,\frac{3}{2}\,;\,\rho^2} \mcom
\]
where $\hypergeometric$ denotes a hypergeometric function with the specified parameters. The proof
of the above identity combines simple tools from Hermite analysis with known integral representations from 
the theory of special functions, and may be useful in other applications of the rounding procedure.

Note that in the Grothendieck case, we have
$\gamma_{p^*}^{p^*} = \gamma_{q}^{q} = \sqrt{2/\pi}$, and the remaining part is simply the
$\sin^{-1}$ function.
In the Krivine rounding scheme, the transformations $\phi$ and $\psi$ are chosen
to satisfy $(2/\pi) \cdot \sin^{-1}\inparen{\ip{\phi(u), \psi(v)}} = c \cdot \ip{u,v}$, where the
constant $c$ then governs the approximation ratio. The transformations $\phi(u)$ and $\psi(v)$ taken
to be of the form $\phi(u) = \oplus_{i=1}^{\infty} a_i \cdot u^{\otimes i}$ such that
\[
\ip{\phi(u), \psi(v)} ~=~ c' \cdot \sin\inparen{\ip{u,v}} \qquad \text{and} \qquad \norm{2}{\phi(u)}
= \norm{\psi(v)} = 1 \mper
\]
If $f$ represents (a normalized version of) the function of $\rho$ occurring in the identity above
(which is $\sin^{-1}$ for hyperplane rounding), then the approximation ratio is governed by the
function $h$ obtained by replacing every Taylor coefficient of $f^{-1}$ by its absolute value. While
$f^{-1}$ is simply the $\sin$ function (and thus $h$ is the $\sinh$ function) in the  Grothendieck
problem, no closed-form expressions are available for general $p$ and $q$.

The task of understanding the approximation ratio thus reduces to the analytic task of understanding
the \emph{family} of the functions $h$ obtained for different values of $p$ and $q$. Concretely, the
approximation ratio is given by the value $1/(h^{-1}(1)\cdot \gamma_{q}\,\gamma_{p^*})$. 
At a high level, we prove bounds on $h^{-1}(1)$ by establishing properties of the Taylor coefficients of 
the family of functions $f^{-1}$, \ie the family given by 
\[
    \inbraces{
    f^{-1}
    ~~|~~
    f(\rho) = \rho\cdot \hypergeometric\!\inparen{a_1,b_1\,;\,3/2\,;\,\rho^2}
    ~,~
    a_1,b_1\in [0,1/2]
    } \mper
\] 
While in the cases considered earlier, the functions $h$ are easy to determine in terms of $f^{-1}$ 
via succinct formulae~\cite{Krivine77, Haagerup81, AN04} or can be truncated after the cubic term
\cite{NR14}, neither of these are true for the family of functions we consider. 
Hypergeometric functions are a rich and expressive class of functions,  capturing many of the 
special functions appearing in Mathematical Physics and various ensembles of orthogonal polynomials. 
Due to this expressive power, the set of inverses is not well understood. 
In particular, while the coefficients of $f$ are monotone in $p$ and $q$, this is not true for
$f^{-1}$. Moreover, the rates of decay of the coefficients may range from inverse polynomial to
super-exponential.
We analyze the coefficients of $f^{-1}$ using complex-analytic methods inspired by (but quite
different from) the work of Haagerup \cite{Haagerup81} on bounding the  complex Grothendieck
constant. 
The key technical challenge in our work is in \emph{arguing systematically about a family 
of inverse hypergeometric functions} which we address by developing methods to estimate the 
values of a family of contour integrals. 

While our method only gives a bound of the form $h^{-1}(1) \geq \sinh^{-1}(1)/(1+\eps_0)$, 
we believe this is an artifact of the analysis and the true bound should
indeed be $h^{-1}(1) \geq \sinh^{-1}(1)$.

\subsection{Relation to Factorization Theory}
Let $X, Y$ be  Banach spaces, and let $A: X \to Y$ be a continuous linear operator. As before,
the norm $\norm{X}{Y}{A}$ is defined as
\[
\norm{X}{Y}{A} ~:=~ \sup_{x \in X \setminus \{0\}} \frac{\norm{Y}{Ax}}{\norm{X}{x}} \mper
\]
The operator $A$ is said to be factorize through Hilbert space if the \defn{factorization constant}
of $A$ defined as
\[
    \factorConst{A} ~:=~ \inf_{H} \inf_{BC=A} \frac{\norm{X}{H}{C}\cdot \norm{H}{Y}{B}}{\norm{X}{Y}{A}}
\]
is bounded, where the infimum is taken over all Hilbert spaces $H$ and all operators $B: H \to Y$
and $C: X \to H$. The factorization gap for spaces $X$ and $Y$ is then defined as
$\factorSpConst{X}{Y} := \sup_A \factorConst{A}$ where the supremum runs over all continuous operators 
$A:X\to Y$. 

The theory of factorization of linear operators is a cornerstone of modern functional analysis and has also 
found many applications outside the field (see \cite{Pisier86, AK06} for more information). 
An application to theoretical computer science was found by Tropp~\cite{Tropp09} who used 
the Grothendieck factorization~\cite{Grothendieck56} to give an algorithmic version of a celebrated 
column subset selection result of Bourgain and Tzafriri~\cite{BT87}. 

As an almost immediate consequence of convex programming duality, our new algorithmic results also 
imply some improved factorization results for $\ell_p^{n},\ell_{q}^{m}$. We first state some classical 
factorization results, for which we will use $T_2(X)$ and $C_2(X)$ 
to respectively denote the Type-2 and Cotype-2 constants of $X$. We refer the interested reader 
to \cref{factorization} for a more detailed description of factorization theory as well as the relevant 
functional analysis preliminaries. 

The \Kwapien-Maurey~\cite{Kwapien72a,Maurey74} theorem states that for any pair of Banach spaces 
$X$ and $Y$
\[
\factorSpConst{X}{Y} ~\leq~ T_2(X) \cdot C_2(Y) \mper
\]
However, Grothendieck's result \cite{Grothendieck56} shows that a much better bound is possible 
in a case where $T_2(X)$ is unbounded. In particular,
\[
\factorSpConst{\ell_{\infty}^n}{\ell_{1}^{m}} ~\leq~ K_G \mcom
\]
for all $m, n \in \N$. Pisier~\cite{Pisier80} showed that if $X$ or $Y$ satisfies the approximation property 
(which is always satisfied by finite-dimensional spaces), then 
\[
\factorSpConst{X}{Y} ~\leq~  \inparen{2\cdot C_2(X^*) \cdot C_2(Y)}^{3/2} \mper
\]
We show that the approximation ratio of Nesterov's relaxation is in fact an upper bound on the
factorization gap for the spaces $\ell_p^n$ and $\ell_q^m$. Combined with our upper bound on the
integrality gap, we show an improved bound on the factorization constant, \ie
for any $p \geq 2 \geq q$ and $m,n \in \N$, we have that for $X = \ell_{p}^n$, $Y=\ell_{q}^m$ 
\[
\factorSpConst{X}{Y} ~\leq~ \frac{1+\eps_0}{\sinh^{-1}(1)} \cdot \inparen{C_2(X^*) \cdot C_2(Y)} \mcom
\]
where $\eps_0 \leq 0.00863$ as before. This improves on Pisier's bound for all $p\geq 2\geq q$, and 
for certain ranges of $(p,q)$ it also improves upon $K_G$ and the bound of \Kwapien-Maurey. 

\subsection{Approximability and Factorizability}
Let $(X_n)$ and $(Y_m)$ be sequences of Banach spaces such that $X_n$ is over the vector space $\R^n$ 
and $Y_m$ is over the vector space $\R^m$. We shall say a pair of sequences $((X_n),(Y_m))$ factorize 
if $\factorSpConst{X_n}{Y_m}$ is bounded by a constant independent of $m$ and $n$. Similarly, 
we shall say a pair of families $((X_n),(Y_m))$ are computationally approximable if there exists a polynomial 
$R(m,n)$, such that for every $m,n \in \N$, there is an algorithm with runtime $R(m,n)$ approximating 
$\norm{X_n}{Y_m}{A}$ within a constant independent of $m$ and $n$ (given an oracle for computing the norms of vectors and a separation oracle for the unit balls of the norms). 
We consider the natural question of characterizing the families of norms that 
are approximable and their connection to factorizability and Cotype. 

The pairs $(p,q)$ (assuming $p^*,q\neq \infty$) for which $(\ell_p^{n},\ell_q^{m})$ is known (resp.\ not known) to factorize, are precisely 
those pairs $(p,q)$ which are known to be computationally approximable (resp.\ inapproximable assuming hardness conjectures like $\textrm{P}\neq\NP$ and \textrm{ETH}). Moreover the Hilbertian case which 
trivially satisfies factorizability, is also known to be computationally approximable (with approximation 
factor 1). 

It is tempting to ask whether the set of computationally approximable pairs is closely related to the set 
of factorizable pairs or the pairs for which $X^*_n,Y_m$ have bounded (independent of $m,n$) Cotype-2 
constant. 
Further yet, is there a connection between the approximation factor and the factorization constant, 
or approximation factor and Cotype-2 constants (of $X_n^*$ and $Y_m$)? 
Our work gives some modest additional evidence towards such conjectures. 
Such a result would give credibility to the appealing intuitive idea of the 
approximation factor being dependent on the ``distance'' to a Hilbert space. 
%


\subsection{Notation}
For a non-negative real number $r$, we define the $r$-th Gaussian norm of a standard gaussian 
$g$ as $\gamma_r\defeq(\Ex{g\sim\gaussian{0}{1}}{\abs{g}^r})^{\nfrac{1}{r}}\mper$

Given a vector $x$, we define the $r$-norm as $\norm{r}{x}^r=\sum_{i}{\abs{x_i}^r}$ for all $r\ge 1$. For any
$r\ge 0$, we denote the dual norm by $r^*$, which satisfies the equality: $\frac{1}{r}+\frac{1}{r^*}=1$.

For $p\ge 2\ge q\ge 1$, we will use the following notation: $a\defeq p^*-1$ and $b\defeq q-1$. We note that
$a,b\in [0,1]$.

For a $m\times n$ matrix $M$ (or vector, when $n=1$). For an unitary function $f$, we define $f[M]$ to 
be the matrix $M$ with entries defined as $(f[M])_{i,j}=f(M_{i,j})$ for $i\in[m],j\in[n]$. For vectors 
$u,v\in\R^{\ell}$, we denote by $u\circ v\in\R^{\ell}$ the entry-wise/Hadamard product of $u$ and $v$. 
We denote the concatenation of two vectors $u$ and $v$ by $u\oplus v$. For a vector $u$, we 
use $\Diag{u}$ to denote the diagonal matrix with the entries of $u$ forming the diagonal, and for 
a matrix $M$ we use $\diag(M)$ to denote the vector of diagonal entries. 

For a function $f(\tau)=\sum_{k\ge 0} f_k\cdot \tau^k$ defined as a power series, we denote the function
$\absolute{f}(\tau)\defeq \sum_{k\ge 0}\abs{f_k}\cdot \tau^k$. 

\section{Analyzing the Approximation Ratio via Rounding}
We will show that $\CP{A}$ is a good approximation to $\norm{p}{q}{A}$ by using an appropriate 
generalization of Krivine's rounding procedure. Before stating the generalized procedure, we shall 
give a more detailed summary of Krivine's procedure. 

\subsection{Krivine's Rounding Procedure}
Krivine's procedure centers around the classical random hyperplane rounding. In this context, we define the 
random hyperplane rounding procedure on an input pair of matrices $U\in \R^{m\times \ell},~ 
V\in\R^{n\times \ell}$ as outputting the vectors $\sgn [U\bfg]$ and $\sgn [V\bfg]$ where $\bfg\in \R^{\ell}$ 
is a vector with i.i.d. standard Gaussian coordinates ($f[v]$ denotes entry-wise application of a scalar 
function $f$ to a vector $v$. We use the same convention for matrices.). 
The so-called Grothendieck identity states that for vectors $u,v\in \R^\ell$, 
\[
    \Ex{\sgn\!{\mysmalldot{\bfg}{u}}\cdot \sgn\!{\mysmalldot{\bfg}{v}}} 
    = 
    \frac{\sin^{-1}\!{\mysmalldot{\widehat{u}}{\widehat{v}}}}{\pi/2}
\]
where $\widehat{u}$ denotes $u/\norm{2}{u}$. This implies the following equality which we will call 
the hyperplane rounding identity: 
\begin{equation}
    \Ex{\sgn [U\bfg](\sgn [V\bfg])^T}
    = 
    \frac{\sin^{-1}[\hatU \hatV^T]}{\pi/2} \mper
\end{equation}
where for a matrix $U$, we use $\hatU$ to denote the matrix obtained by replacing the rows of $U$ 
by the corresponding unit (in $\ell_2$ norm) vectors. 
Krivine's main observation is that for any matrices $U,V$, there exist matrices 
$\phi(\hatU),\psi(\hatV)$ with unit vectors as rows, such that 
\[
    \phi(\hatU)\,\psi(\hatV)^T = \sin [(\pi/2) \cdot c\cdot \hatU\hatV^T]
\]
where $c = \sinh^{-1}(1)\cdot 2/\pi$. Taking $\hatU,\hatV$ to be the optimal solution to $\CP{A}$, 
it follows that 
\[
    \norm{\infty}{1}{A}
    \geq 
    \mydot{A}{\Ex{\sgn [\phi(\hatU)\,\bfg]~(\sgn [\psi(\hatV)\,\bfg])^T}} 
    = 
    \mysmalldot{A}{c\cdot \hatU\hatV^T}
    = 
    c\cdot \CP{A} \mper
\]
The proof of Krivine's observation follows from simulating the Taylor series of a scalar function using inner 
products. We will now describe this more concretely. 
\begin{observation}[Krivine]
    \label[observation]{simulating:taylor}
    Let $f:[-1,1]\to\R$ be a scalar function 
    satisfying $f(\rho) = \sum_{k\geq 1} f_k\, \rho^k$ for an absolutely convergent series $(f_k)$. 
    Let $\absolute{f}(\rho):= \sum_{k\geq 1} |f_k|\,\rho^{k}$ and further for vectors 
    $u,v\in\R^{\ell}$ of $\ell_2$-length at most $1$, let 
    \begin{align*}
        &\seriesLeft{f}{u} := (\sgn(f_1)\sqrt{f_1}\cdot u)\oplus 
        (\sgn(f_2)\sqrt{f_2}\cdot u^{\otimes 2})\oplus 
        (\sgn(f_3)\sqrt{f_3}\cdot u^{\otimes 3})\oplus \cdots \\
        &\seriesRight{f}{v} := (\sqrt{f_1}\cdot v)\oplus (\sqrt{f_2}\cdot v^{\otimes 2})\oplus 
        (\sqrt{f_3}\cdot v^{\otimes 3})\oplus \cdots 
    \end{align*}
    Then for any $U\in \R^{m\times \ell},~V\in \R^{n\times \ell}$, ~$\seriesLeft{f}{\sqrt{c_f}\cdot\hatU}$ 
    and $\seriesRight{f}{\sqrt{c_f}\cdot \hatV}$ have $\ell_2$-unit vectors as rows, and 
    \[
        \seriesLeft{f}{\sqrt{c_f}\cdot\hatU}~\seriesRight{f}{\sqrt{c_f}\cdot \hatV}^T 
        = 
        f\,[c_f\cdot \hatU\hatV^T]
    \]
    where $\seriesLeft{f}{W}$ for a matrix $W$, is applied to row-wise and $c_f := (\absolute{f}^{-1})(1)$. 
\end{observation}

\begin{proof}
    Using the facts $\mysmalldot{y^1\otimes y^2}{y^3\otimes y^4} = 
    \mysmalldot{y^1}{y^3}\cdot \mysmalldot{y^2}{y^4}$ and \\
    $\mysmalldot{y^1\oplus y^2}{y^3\oplus y^4} = 
    \mysmalldot{y^1}{y^3}+ \mysmalldot{y^2}{y^4}$,~we have  
    \begin{itemize}
        \item $\mysmalldot{\seriesLeft{f}{u}}{\seriesRight{f}{v}} = f(\mysmalldot{u}{v})$
        
        \item $\norm{2}{\seriesLeft{f}{u}} = \sqrt{\absolute{f}(\norm{2}{u}^{2})}$
        
        \item $\norm{2}{\seriesRight{f}{v}} = \sqrt{\absolute{f}(\norm{2}{v}^{2})}$ 
    \end{itemize}
    The claim follows. 
\end{proof}

Before stating our full rounding procedure, we first discuss a natural generalization of random hyperplane 
rounding, and much like in Krivine's case this will guide the final procedure. 

\subsection{Generalizing Random Hyperplane Rounding -- \Holder Dual Rounding}
Fix any convex bodies $B_1 \subset \R^m$ and $B_2\subset \R^k$. Suppose that we would like a strategy 
that for given vectors $y\in\R^m,~x\in\R^n$, outputs $\bary\in B_1,~\barx\in B_2$ so that 
$y^TA\,x = \mysmalldot{A}{y\,x^T}$ is close to $\mysmalldot{A}{\bary\, \barx^T}$ for all $A$. A natural 
strategy is to take 
\[
    (\bary,\barx) := \argmax_{(\tildey,\tildex)\in B_1\times B_2} 
    \mydot{\tildey\,\tildex^T}{y\,x^T}
    = 
    \inparen{
    \argmax_{\tildey\in B_1} 
    \mydot{\tildey}{y}
    ~,~
    \argmax_{\tildex\in B_2} 
    \mydot{\tildex}{x}
    }
\]
In the special case where $B$ is the unit $\ell_p$ ball, there is a closed form for an optimal 
solution to $\max_{\tildex\in B} \mysmalldot{\tildex}{x}$, given by $\holderdual{p^*}{x}/
\norm{p^*}{x}^{p^*-1}$, where $\holderdual{p^*}{x} := \sgn[x]\circ |[x]|^{p^*-1}$. Note that for 
$p= \infty$, this strategy recovers the random hyperplane rounding procedure. We shall call this 
procedure, \emph{Gaussian \Holder Dual Rounding} or \HD for short. 

Just like earlier, we will first understand the effect of \HD on a solution pair $U,V$.
For $\rho\in [-1, 1]$, let $\bfg_1 \!\sim_{\rho} \bfg_2$ denote $\rho$-correlated standard 
Gaussians, \ie $\bfg_1 = \rho\,\bfg_2 + \sqrt{1-\rho^2}\,\bfg_3$ ~where 
$(\bfg_2,\bfg_3) \sim \gaussian{0}{\id_2}$, and let 
\[
    \fplain{a}{b}{\rho} 
    :=  
    \Ex{\bfg_1 \sim_\rho \bfg_2}{\sgn(\bfg_1)|\bfg_1|^{b}\sgn(\bfg_2)|\bfg_1|^{a}}
\]
We will work towards a better understanding of $\fplain{a}{b}{\cdot}$ in later sections. 
For now note that we have for vectors $u,v\in\R^{\ell}$, 
\[
    \Ex{\sgn\!{\mysmalldot{\bfg}{u}}\,|\mysmalldot{\bfg}{u}|^{b} \cdot 
    \sgn\!{\mysmalldot{\bfg}{v}} \,|\mysmalldot{\bfg}{v}|^{a}} 
    = 
    \norm{2}{u}^{b}\cdot\norm{2}{v}^{a}\cdot \fplain{a}{b}{\mysmalldot{\widehat{u}}{\widehat{v}\,}} \mper
\]
Thus given matrices $U,V$, we obtain the following generalization of the hyperplane rounding identity 
for \HD: 
\begin{equation}
\label[equation]{cvgp:identity}
    \Ex{\holderdual{q}{[U\bfg]}\,\holderdual{p^*}{[V\bfg]}^T}
    = 
    \Diag{(\norm{2}{u^i}^{b})_{i\in [m]}} \cdot \fplain{a}{b}{[\hatU \hatV^T]}\cdot 
    \Diag{(\norm{2}{v^j}^{a})_{j\in [n]}} 
    \mper
\end{equation}

\subsection{Generalized Krivine Transformation and the Full Rounding Procedure}
\label[subsection]{rounding}

We are finally ready to state the generalized version of Krivine's algorithm. At a high level the algorithm 
simply applies \HD to a transformed version of the optimal convex program solution pair $U,V$. 
Analogous to Krivine's algorithm, the transformation is a type of ``inverse'' of \cref{cvgp:identity}. 

\begin{enumerate}[({Inversion }1)]
    \item Let $(U,V)$ be the optimal solution to $\CP{A}$, and let $(u^i)_{i\in [m]}$ and $(v^j)_{j\in[n]}$ 
    respectively denote the rows of $U$ and $V$. 

    \item Let $c_{a,b} := \inparen{\absolute{\alfin{a}{b}}}^{-1}\!\!(1)$ and let 
        \begin{align*}
            \phi(U) &:= \Diag{(\norm{2}{u^i}^{1/b})_{i\in [m]}}\,\,\seriesLeft{\alfin{a}{b}}{\sqrt{c_{a,b}}\cdot\hatU} 
            \mcom \\
            \psi(V) &:= \Diag{(\norm{2}{v^j}^{1/a})_{j\in [n]}}\,\,\seriesRight{\alfin{a}{b}}{\sqrt{c_{a,b}}\cdot\hatV}
            \mper
        \end{align*}
\end{enumerate}
\begin{enumerate}[({\Holder-Dual }1)]
    \item Let $\bfg \sim \gaussian{0}{\id}$ be an infinite dimensional i.i.d. Gaussian vector. 
    
    \item Return $y:=\holderdual{q}{\phi(U)\,\bfg}/\norm{q}{\phi(U)\,\bfg}^{b}$ and 
    $x:=\holderdual{p^*}{\psi(V)\,\bfg}/\norm{p^*}{\psi(V)\,\bfg}^{a}$. 
\end{enumerate}

\bigskip

\begin{remark}
\label[remark]{round:feasibility}
    Note that $\norm{r^*}{\holderdual{r}{\barx}} = \norm{r}{\barx}^{r-1}$ and so the returned solution pair 
    always lie on the unit $\ell_{q^*}$ and $\ell_{p}$ spheres respectively. 
\end{remark}

\begin{remark}
    Like in~\cite{AN04} the procedure above can be made algorithmic by observing that there always exist 
    $\phi'(U)\in \R^{m\times (m+n)}$ and $\psi'(V)\in \R^{m \times (m+n)}$, whose rows have the exact 
    same lengths and pairwise inner products as those of $\phi(U)$ and $\psi(V)$ above. 
    Moreover they can be computed without explicitly computing $\phi(U)$ and $\psi(V)$ by 
    obtaining the Gram decomposition of 
    \[
        M~:=~
        \left[
        \begin{array}{cc}
        \absolute{\alfin{a}{b}}[c_{a,b}\cdot \hatV\hatV^T] & \fin{a}{b}{[c_{a,b}\cdot \hatU\hatV^T]} \\
        \fin{a}{b}{[c_{a,b}\cdot \hatV\hatU^T]} & \absolute{\alfin{a}{b}}[c_{a,b}\cdot \hatV\hatV^T]
        \end{array}
        \right] \mcom
    \]
    and normalizing the rows of the decomposition according to the definition of $\phi(\cdot)$ and 
    $\psi(\cdot)$ above. The entries of $M$ can be computed in polynomial time with exponentially 
    (in $m$ and $n$) good accuracy by implementing the Taylor series of $\alfin{a}{b}$ 
    upto $\mathrm{poly}(m,n)$ terms (Taylor series inversion can be done upto $k$ terms in time 
    $\mathrm{poly}(k)$).
\end{remark}

\begin{remark}
\label[remark]{algo:rmk}
    Note that the $2$-norm of the $i$-th row (resp. $j$-th row) of $\phi(U)$ (resp. $\psi(V)$) is 
    $\norm{2}{u^i}^{1/b}$ (resp. $\norm{2}{v^j}^{1/a}$).
\end{remark}

We commence the analysis by defining some convenient normalized functions and we will 
also show that $c_{a,b}$ above is well-defined. 
\subsection{Auxiliary Functions}
Let ~$\nf{p}{q}{\rho} := \fplain{p}{q}{\rho}/(\gamma_{p^*}^{p^*}\,\gamma_{q}^{q})$, ~ 
$\alfinabs{a}{b} := \absolute{\alfin{a}{b}} $,~ and ~$\alnh{a}{b} := \absolute{\alnfin{a}{b}}$. 
Also note that $\nhin{a}{b}{\rho} = \hin{a}{b}{\rho}/(\gamma_{p^*}^{p^*}\,\gamma_{q}^{q})$. 

\paragraph{Well Definedness. }
By \cref{inv:coeff:bound}, $\nfin{a}{b}{\rho}$ and $\nh{a}{b}{\rho}$ are well defined for $\rho\in [-1,1]$.  
By (M1) in \cref{monotonicity:properties}, ~
$\kfinc{1} = 1$ and hence $\nh{a}{b}{1} \geq 1$ and $\nh{a}{b}{-1}\leq -1$. Combining this with the 
fact that $\nh{a}{b}{\rho}$ is continuous and strictly increasing on $[-1,1]$, implies that 
$\nhin{a}{b}{x}$ is well defined on $[-1,1]$. 

\medskip
We can now proceed with the analysis. 

\subsection[Bound on Approximation Factor]{$1/(\nhin{p}{q}{1}\cdot \gamma_{p^*}\,\gamma_q)$
Bound on Approximation Factor}

For any vector random variable $\bfX$ in a universe $\Omega$, and scalar valued functions 
$f_1:\Omega \to\R$ and $f_2:\Omega\to (0,\infty)$. Let $\lambda = \Ex{f_1(\bfX)}/\Ex{f_2(\bfX)}$. 
Now we have 
\begin{align*}
    &\max_{x\in\Omega} f_1(x)-\lambda\cdot f_2(x) \geq \Ex{f_1(\bfX)-\lambda\cdot f_2(\bfX)} = 0 \\
    \Rightarrow\quad  
    &\max_{x\in\Omega} f_1(x)/f_2(x) \geq \lambda = \Ex{f_1(\bfX)}/\Ex{f_2(\bfX)}\mper
\end{align*}
Thus we have 
\[
    \norm{p}{q}{A} 
    ~\geq ~
    \frac{\Ex{\mysmalldot{A}{\holderdual{q}{\phi(U)\,\bfg}~\holderdual{p^*}{\psi(V)\,\bfg}^T}}} 
    {\Ex{\norm{q^*}{\holderdual{q}{\phi(U)\,\bfg}}\cdot \norm{p}{\holderdual{p^*}{\psi(V)\,\bfg}}}} 
    ~= ~
    \frac{\mysmalldot{A}{\Ex{\holderdual{q}{\phi(U)\,\bfg}~\holderdual{p^*}{\psi(V)\,\bfg}^T}}} 
    {\Ex{\norm{q^*}{\holderdual{q}{\phi(U)\,\bfg}}\cdot \norm{p}{\holderdual{p^*}{\psi(V)\,\bfg}}}} 
    \mcom
\]
which allows us to consider the numerator and denominator separately. 
We begin by proving the equality that the above algorithm was designed to satisfy: 
\begin{lemma}
    \label[lemma]{numerator}
    $
        \Ex{\holderdual{q}{\phi(U)\,\bfg}~ \holderdual{p^*}{\psi(V)\,\bfg}^T} 
        ~=~ 
        c_{a,b}\cdot (\tildeU \tildeV^T)
    $
\end{lemma}

\begin{proof}
    \begin{align*}
        &\quad~~\Ex{\holderdual{q}{\phi(U)\,\bfg}~ \holderdual{p^*}{\psi(V)\,\bfg}^T}  \\
        &= 
        \Diag{(\norm{2}{u^i})_{i\in [m]}} \cdot 
        \fplain{a}{b}{[\seriesLeft{\alfin{a}{b}}{\sqrt{c_{a,b}}\cdot\hatU}\cdot
        \seriesRight{\alfin{a}{b}}{\sqrt{c_{a,b}}\cdot\hatV}^T]} \cdot 
        \Diag{(\norm{2}{v^j})_{j\in [n]}} \\
        &~~(\text{by \cref{cvgp:identity} and \cref{algo:rmk}}) \\
        &= 
        \Diag{(\norm{2}{u^i})_{i\in [m]}} \cdot 
        \fplain{a}{b}{[\fin{a}{b}{[c_{a,b}\cdot \hatU\hatV^T]}]} \cdot 
        \Diag{(\norm{2}{v^j})_{j\in [n]}} \\
        &~~(\text{by \cref{simulating:taylor}}) \\
        &= 
        \Diag{(\norm{2}{u^i})_{i\in [m]}} \cdot 
        c_{a,b}\cdot \hatU\hatV^T \cdot 
        \Diag{(\norm{2}{v^j})_{j\in [n]}} \\
        &= 
        c_{a,b}\cdot UV^T \qedforce
    \end{align*}
\let\qed\relax
\end{proof}

It remains to upper bound the denominator which we do using a straightforward convexity argument. 
\begin{lemma}
    \label[lemma]{denominator}
    $
        \Ex{\norm{q}{\phi(U)\,\bfg}^{b}\cdot \norm{p^*}{\psi(V)\,\bfg}^{a}}
        \leq  
        \gamma_{p^*}^{a}\, \gamma_{q}^{b} \mper
    $
\end{lemma}

\begin{proof}
    \begin{align*}
        &\quad~\Ex{\norm{q}{\phi(U)\,\bfg}^{b}\cdot \norm{p^*}{\psi(V)\,\bfg}^{a}} \\
        &\leq~ 
        \Ex{\norm{q}{\phi(U)\,\bfg}^{q^* b}}\!^{1/q^*} \cdot \Ex{\norm{p^*}{\psi(V)\,\bfg}^{p a}}\!^{1/p}
        &&\inparen{\frac{1}{p}+\frac{1}{q^*}\leq 1} \\
        &=~ 
        \Ex{\norm{q}{\phi(U)\,\bfg}^{q}}\!^{1/q^*} \cdot \Ex{\norm{p^*}{\psi(V)\,\bfg}^{p^*}}\!^{1/p} \\
        &=~ 
        \insquare{\sum_{i\in [m]}\Ex{|\gaussian{0}{\norm{2}{u^i}^{1/b}}|^{q}}}^{1/q^*} 
        \cdot 
        \insquare{\sum_{j\in [n]}\Ex{|\gaussian{0}{\norm{2}{v^j}^{1/a}}|^{p^*}}}^{1/p}  
        &&(\text{By \cref{algo:rmk}}) \\
        &=~ 
        \insquare{\sum_{i\in [m]} \norm{2}{u^i}^{q/b}}^{1/q^*} 
        \cdot 
        \insquare{\sum_{j\in [n]} \norm{2}{v^j}^{p^*/a}}^{1/p}  
        \cdot \gamma_{q}^{q/q^*}\,\gamma_{p^*}^{p^*/p} \\
        &=~ 
        \insquare{\sum_{i\in [m]} \norm{2}{u^i}^{q^*}}^{1/q^*} 
        \cdot 
        \insquare{\sum_{j\in [n]} \norm{2}{v^j}^{p}}^{1/p}  
        \cdot \gamma_{q}^{b}\,\gamma_{p^*}^{a} \\
        &=~
        \gamma_{q}^{b}\,\gamma_{p^*}^{a} 
        &&\hspace{-40 pt}(\text{feasibility of } U,V) \qedhere
    \end{align*}
\end{proof}

\medskip

We are now ready to prove our approximation guarantee. 
\begin{lemma}
    \label[lemma]{defect:bound:implies:apx}
    Consider any $1\leq q\leq 2\leq p \leq \infty$. Then, 
    \[
        \frac{\CP{A}}{\norm{p}{q}{A}} 
        ~\leq ~ 
        1/(\gamma_{p^*}\,\gamma_q \cdot \nhin{a}{b}{1}) 
    \]
\end{lemma}

\begin{proof}
    \begin{align*}
        \norm{p}{q}{A}~
        &\geq ~
        \frac{\mysmalldot{A}{\Ex{\holderdual{q}{\phi(U)\,\bfg}~\holderdual{p^*}{\psi(V)\,\bfg}^T}}}
        {\Ex{\norm{q^*}{\holderdual{q}{\phi(U)\,\bfg}}\cdot \norm{p}{\holderdual{p^*}{\psi(V)\,\bfg}}}} \\
        &= ~
        \frac{\mysmalldot{A}{\Ex{\holderdual{q}{\phi(U)\,\bfg}~\holderdual{p^*}{\psi(V)\,\bfg}^T}}}
        {\Ex{\norm{q}{\phi(U)\,\bfg}^{b}\cdot \norm{p^*}{\psi(V)\,\bfg}^{a}}} 
        &&(\text{by \cref{round:feasibility}}) \\
        &= ~ 
        \frac{c_{a,b}\cdot \mysmalldot{A}{UV^T}}
        {\Ex{\norm{q}{\phi(U)\,\bfg}^{b}\cdot \norm{p^*}{\psi(V)\,\bfg}^{a}}} 
        &&(\text{by \cref{numerator}}) \\
        &= ~ 
        \frac{c_{a,b}\cdot \CP{A}}
        {\Ex{\norm{q}{\phi(U)\,\bfg}^{b}\cdot \norm{p^*}{\psi(V)\,\bfg}^{a}}} 
        &&(\text{by optimality of }U,V) \\
        &\geq ~ 
        \frac{c_{a,b}\cdot \CP{A}}{\gamma_{p^*}^{a}\,\gamma_q^{b}} 
        &&(\text{by \cref{denominator}}) \\
        &= ~ 
        \frac{\hin{a}{b}{1}\cdot \CP{A}}{\gamma_{p^*}^{a}\,\gamma_q^{b}} \\
        &= ~ 
        \nhin{a}{b}{1}\cdot \gamma_{p^*}\,\gamma_{q}\cdot \CP{A}  \qedforce
    \end{align*}
\let\qed\relax
\end{proof}

We next begin the primary technical undertaking of this paper, namely proving upper bounds on 
$\nhin{p}{q}{1}$. 


\section[Hypergeometric Representation]{Hypergeometric Representation of  $\nf{a}{b}{x}$}

In this section, we show that $\nf{a}{b}{\rho}$ can be represented using the Gaussian hypergeometric 
function $\hypergeometric$. The result of this section can be thought of as a generalization of the 
so-called Grothendieck identity for hyperplane rounding which simply states that 
\[
    \nf{0}{0}{\rho}~=~\frac{\pi}{2}\cdot \Ex{\bfg_1 \sim_\rho \,\bfg_2}{\sgn(\bfg_1)\sgn(\bfg_2)}~=~\sin^{-1}
    (\rho)
\]
We believe the result of this section and its proof technique to be of independent 
interest in analyzing generalizations of hyperplane rounding to convex bodies other than the hypercube. 

Recall that $\fplain{a}{b}{\rho}$ is defined as follows: 
\[
    \Ex{\bfg_1 \sim_\rho \,\bfg_2}{\sgn(\bfg_1)|\bfg_1|^{a}\sgn(\bfg_2)|\bfg_1|^{b}}
\]
where $a=p^*-1$ and $b=q-1$. 
Our starting point is the simple observation that the above expectation can be viewed as the noise 
correlation (under the Gaussian measure) of the functions  $\fa{\tau}:= \sgn{\tau}\cdot |\tau|^{a}$ and 
$\fb{\tau}:= \sgn{\tau}\cdot |\tau|^{b}$. Elementary Hermite analysis then implies that it suffices 
to understand the Hermite coefficients of $\alfa$ and $\alfb$ individually, in order to understand the 
Taylor coefficients of $\alnf{a}{b}$. To understand the Hermite coefficients of $\alfa$ and $\alfb$ individually, 
we use a generating function approach. More specifically, we derive an integral representation for 
the generating function of the (appropriately normalized) Hermite coefficients which fortunately turns out 
to be closely related to a well studied special function called the parabolic cylinder function. 

Before proceeding, we require some preliminaries. 

\subsection[Hermite Preliminaries]{Hermite Analysis Preliminaries}
Let $\gamma$ denote the standard Gaussian probability distribution. 
For this section (and only for this section), the (Gaussian) inner product for functions
$f,h\in (\R,\gamma) \rightarrow \R$ is defined as 
\[
    \mysmalldot{f}{h}:=\int_{\R} f(\tau)\cdot h(\tau)\,\,d\gamma(\tau)
    =\Ex{\tau\sim \gaussian{0}{1}}{f(\tau)\cdot h(\tau)}\mper
\]
Under this inner product there is a complete set of orthonormal polynomials $(H_k)_{k\in \N}$ defined below. 
\begin{definition}
    \label[def-hermite]{def:hermite}
    For a natural number $k$, then the $k$-th \emph{Hermite} polynomial $H_k: \R\to\R$
    \[
        H_k(\tau)=\frac{1}{\sqrt{k!}}\cdot (-1)^{\,k} \cdot \ee^{\tau^2/2}\cdot 
        \frac{d^k}{d\tau^k}\,\ee^{-\tau^2/2} \mper
    \]
\end{definition}
\noindent
Any function $f$ satisfying $\int_{\R}|f(\tau)|^{2}\,d\gamma(\tau)<\infty$ has a Hermite expansion given by 
$
    f = \sum_{k\geq 0} \widehat{f}_{k} \cdot H_{k}
$
where 
$
    \widehat{f}_k = \mysmalldot{f}{H_k} \mper
$
\medskip

\noindent
We have 
\begin{fact}
\label[fact]{even:odd:hermite}
    $H_k(\tau)$ is an even (resp. odd) function when $k$ is even (resp. odd). 
\end{fact}
We also have the Plancherel Identity (as Hermite polynomials form an orthonormal basis):
\begin{fact}
    \label[fact]{plancherel-identity}
    For two real valued functions $f$ and $h$ with Hermite coefficients $\widehat{f}_k$ and
    $\widehat{h}_k$, respectively, we have:
    \[
        \mysmalldot{f}{h} = \sum_{k\geq 0} \widehat{f}_k\cdot \widehat{h}_k \mper
    \]
\end{fact}
The generating function of appropriately normalized Hermite polynomials satisfies the following identity: 
\begin{equation}
\label[equation]{hermite:generating:function}
    e^{\,\tau\lambda - \lambda^2/2} 
    = 
    \sum_{k\geq 0} H_k(\tau)\cdot \frac{\lambda^{k}}{\sqrt{k!}} \mper
\end{equation}

Similar to the noise operator in Fourier analysis, we define the corresponding noise operator $\noise{}$
for Hermite analysis: 
\begin{definition}
    For $\rho \in [-1,1]$ and a real valued function $f$, we define the function $\noise{f}$ as:
    \[
        (\noise{f})(\tau) = \int_{\R} f\inparen{\rho\cdot \tau + \sqrt{1-\rho^2}\cdot
        \theta}\,d\gamma(\theta) =\Ex{\tau'\sim_\rho\,\tau}{f(\tau')}\mper
    \]
    \label[definition]{noise-operator}
\end{definition}
Again similar to the case of Fourier analysis, the Hermite coefficients admit the following identity:
\begin{fact}
    \label[fact]{hermite-noise}
    $
        \widehat{(\noise{f})}_k =  \rho^{k} \cdot \widehat{f}_k\mper
    $
\end{fact}

We recall that the $\fplain{a}{b}{\rho}=\Ex{\bfg_1\sim\rho\,\bfg_2}{\alfa(\bfg_1)\cdot \alfb(\bfg_2))}$,
where $\fwild{c}{\tau}:=\sgn(\tau)\cdot \abs{\tau}^c$ for $c\in \{a,b\}$. 
As mentioned at the start of the section, we now note that $\nf{a}{b}{\rho}$ is the noise correlation 
of $\alfa$ and $\alfb$. Thus we can relate the Taylor coefficients of $\nf{a}{b}{\rho}$, to the Hermite 
coefficients of $\alfa$ and $\alfb$. 
\begin{claim}[Coefficients of $\fplain{a}{b}{\rho}$]
    \label[claim]{noise:correlation} For $\rho\in[-1,1]$, we have:
    \[
        \fplain{a}{b}{\rho}=\sum_{k\ge 0}\rho^{2k+1}\cdot \hfac{2k+1}\cdot \hfbc{2k+1}\mcom
    \] where $\hfac{i}$ and $\hfbc{j}$ are the $i$-th and $j$-th Hermite coefficients of
    $\alfa$ and $\alfb$, respectively. Moreover, $\hfac{2k}=\hfbc{2k}=0$ for $k\ge 0$.
\end{claim}
\begin{proof}
    We observe that both $\alfa$ and $\alfb$ are odd functions and hence \cref{even:odd:hermite} implies that
    $\hfac{2k}=\hfbc{2k}=0$ for all $k\ge 0$ -- as $\alfa(\tau)\cdot H_{2k}(\tau)$ is an odd function of
    $\tau$.
    \begin{align*}
        \fplain{a}{b}{\rho} &=  \Ex{\bfg_1\sim\rho\,\bfg_2}{\alfa(\bfg_1)\cdot \alfb(\bfg_2))}\\
        &= \Ex{\bfg_1}{\alfa(\bfg_1)\cdot \noise{\alfb}(\bfg_1)} & (\text{\cref{noise-operator}})\\
        &= \mysmalldot{\alfa}{\noise{\alfb}}\\
    &= \sum_{k\ge 0} \hfac{k}\cdot \widehat{(\noise{\alfb})}_{k} & (\text{\cref{plancherel-identity}})\\
        &= \sum_{k\ge 0} \hfac{2k+1}\cdot \widehat{(\noise{\alfb})}_{2k+1} \\
        &= \sum_{k\ge 0} \rho^{2k+1}\cdot \hfac{2k+1}\cdot \hfbc{2k+1} &
        (\text{\cref{hermite-noise}})\mper \qedforce
    \end{align*}
\let\qed\relax
\end{proof}

\subsection[Hermite Coefficients via Parabolic Cylinder Functions]
{Hermite Coefficients of $\alfa$ and $\alfb$ via Parabolic Cylinder Functions}
In this subsection, we use the generating function of Hermite polynomials to to obtain an integral
representation for the generating function of the ($\sqrt{k!}$ normalized) odd Hermite coefficients
of $\alfa$ (and similarly of $\alfb$) is closely related to 
a special function called the parabolic cylinder function. We then use known facts about the 
relation between parabolic cylinder functions and confluent hypergeometric functions, to show 
that the Hermite coefficients of $\alfwild{c}$ can be obtained from the Taylor coefficients of a 
confluent hypergeometric function.

Before we state and prove the main results of this subsection we need some preliminaries:

\subsubsection{Gamma, Hypergeometric and Parabolic Cylinder Function Preliminaries}
For a natural number $k$ and a real number $\tau$, we denote the rising factorial as
$(\tau)_k\defeq \tau \cdot (\tau+1)\cdot \cdots (\tau+k-1)$.
We now define the following fairly general classes of functions and we later use them we obtain a
Taylor series representation of $\fplain{a}{b}{\tau}$.
\begin{definition}
    \label[definition]{conf-hypergeometric}
    The confluent hypergeometric function with parameters $\alpha,\beta$, and $\lambda$ as:
    \[
        \confHG(\alpha\,;\beta\,;\lambda)\defeq \sum_{k} \frac{(\alpha)_k}{(\beta)_k}\cdot \frac{\lambda^k}{k!}\mper
    \]
\end{definition}
The (Gaussian) hypergeometric function is defined as follows:
\begin{definition}
    \label[definition]{hypergeometric}
    The hypergeometric function with parameters $w,\alpha,\beta$ and $\lambda$ as:
    \[
        \hypergeometric(w,\alpha\,;\beta\,;\lambda) \defeq \sum_k \frac{(w)_k\cdot (\alpha)_k}{(\beta)_k}\cdot
        \frac{\lambda^k}{k!}\mper
    \]
\end{definition}

    Next we define the $\Gamma$ function:
    \begin{definition}
	    For a real number $\tau$, we define:
	    \[
	        \Gamma(\tau)\defeq \int_{0}^{\infty}t^{\tau-1}\cdot \ee^{-t}\,dt\mper
	    \]
	    \label[definition]{Gamma-function}
	\end{definition} The $\Gamma$ function has the following property:
    \begin{fact}[Duplication Formula]
    \label[fact]{duplication:formula}
        \[
            \frac{\Gamma(2\tau)}{\Gamma(\tau)} = \frac{\Gamma(\tau+1/2)}{2^{1-2\tau}\sqrt{\pi}}
        \]
    \end{fact}
	We also note the relationship between $\Gamma$ and $\gamma_r$:
	\begin{fact}
	\label[fact]{gamma:and:Gamma}
	    For $r\in [0,\infty)$, 
	    \[
	        \gamma_r^{r}~:=~\Ex{\bfg\sim \gaussian{0}{1}}{|\bfg|^r}~=~\frac{2^{r/2}}{\sqrt{\pi}}\cdot
	        \Gamma\inparen{\frac{1+r}{2}} \mper
	    \]
	\end{fact}
	\begin{proof}
	    \begin{align*}
	        \Ex{\bfg\sim\gaussian{0}{1}}{\abs{\bfg}^r} &= \frac{\sqrt{2}}{\sqrt{\pi}}\cdot\int_{0}^{\infty}
	        \abs{\bfg}^r\cdot\ee^{-\nfrac{\bfg^2}{2}}\,d\bfg\\
	        &= \sqrt{\frac{2}{\pi}}\cdot 2^{(r-1)/2}\cdot\int_{0}^{\infty}\abs{\frac{\bfg^2}{2}
	        }^{(r-1)/2}\!\cdot\ee^{-\nfrac{\bfg^2}{2}}\!\cdot \bfg\,d\bfg\\
	        &= \frac{2^{r/2}}{\sqrt{\pi}}\cdot\Gamma\inparen{\frac{1+r}{2}} \qedforce
	    \end{align*} 
    \let\qed\relax
	\end{proof}

Next, we record some facts about parabolic cylinder functions:
\begin{fact}[12.5.1 of \cite{Lozier03}]
    \label[fact]{integral-parabolic}
    Let $U$ be the function defined as 
    \[
        U(\alpha,\lambda)\defeq\frac{\ee^{\lambda^2/4}}{\Gamma\inparen{\frac{1}{2}+\alpha}}
        \int_{0}^{\infty} t^{\alpha-1/2}\cdot \ee^{-(t+\lambda)^2/2} \,dt\mcom
    \]for all $\alpha$ such that $\Re(\alpha)>-\frac{1}{2}\mper$
    The function $U(\alpha,\pm \lambda)$ is a parabolic cylinder function and is a standard
    solution to the differential equation: $\frac{d^2 w}{d\lambda^2}-\inparen{\frac{\lambda^2}{4}+\alpha}w=0$. 
\end{fact}
Next we quote the confluent hypergeometric representation of the parabolic cylinder function $U$ 
defined above:
\begin{fact}[12.4.1, 12.2.6, 12.2.7, 12.7.12, and 12.7.13 of \cite{Lozier03}]
    \label[fact]{power-parabolic}
    \begin{align*}
        U(\alpha,\lambda)=\frac{\sqrt{\pi}}{2^{\alpha/2+1/4}\cdot
        \Gamma\inparen{\frac{3}{4}+\frac{\alpha}{2}}}\cdot
        \ee^{\lambda^2/4}\cdot\confHG\inparen{-\frac{1}{2}\alpha+\frac{1}{4}\,;\,
        \frac{1}{2}\,;\,-\frac{\lambda^2}{2}}\\
        -\frac{\sqrt{\pi}}{2^{\alpha/2-1/4}\cdot
        \Gamma\inparen{\frac{1}{4}+\frac{\alpha}{2}}}\cdot\lambda\cdot\ee^{\lambda^2/4}
        \cdot\confHG\inparen{-\frac{\alpha}{2}+\frac{3}{4}\,;\,\frac{3}{2}\,;\,-\frac{\lambda^2}{2}}
    \end{align*}
\end{fact}
Combining the previous two facts, we get the following:
\begin{corollary}
    \label[corollary]{parabolic-conf-hypergeometric}
    For all real $\alpha>-\frac{1}{2}$, we have:
    \begin{align*} 
        \int_{0}^{\infty} t^{\alpha-1/2}\cdot \ee^{-(t+\lambda)^2/2} \,dt ~~=~
        \frac{\sqrt{\pi}\cdot\Gamma\inparen{\frac{1}{2}+\alpha}}{2^{\alpha/2+1/4}
        \cdot\Gamma\inparen{\frac{3}{4}+\frac{\alpha}{2}}}
        \cdot\confHG\inparen{-\frac{\alpha}{2}+\frac{1}{4}\,;\,\frac{1}{2}\,;\,-\frac{\lambda^2}{2}}
        \\-\frac{\sqrt{\pi}\cdot\Gamma\inparen{\frac{1}{2}+\alpha}}{2^{\alpha/2-1/4}
        \cdot\Gamma\inparen{\frac{1}{4}+\frac{\alpha}{2}}}
        \cdot \lambda\cdot\confHG\inparen{-\frac{\alpha}{2}+\frac{3}{4}\,;\,\frac{3}{2}\,;\,-\frac{\lambda^2}{2}}
        \mper
    \end{align*}
\end{corollary}
\subsubsection{Generating Function of Hermite Coefficients and its Confluent Hypergeometric Representation}
Using the generating function of (appropriately normalized) Hermite polynomials, we derive an integral 
representation for the generating function of the (appropriately normalized) Hermite coefficients of 
$\alfa$ (and similarly $\alfb$):
\begin{lemma}
    \label[lemma]{gen:function:integral:rep}
    For $c\in\{a,b\}$, let $\hfwildc{c}{k}$ denote the $k$-th Hermite coefficient of 
    $\fwild{c}{\tau} := \sgn{(\tau)}\cdot |\tau|^{c}$. 
    Then we have the following identity:
    \[
        \sum_{k\geq 0} \frac{\lambda^{2k+1}}{\sqrt{(2k+1)!}}\cdot \hfwildc{c}{2k+1}=
        \frac{1}{\sqrt{2\pi}}\int_{0}^{\infty} \tau^{c}\cdot 
        \inparen{\ee^{-(\tau-\lambda)^2/2} - \ee^{-(\tau+\lambda)^2/2}} \,d\tau \mper
    \] 
\end{lemma}
\begin{proof}
    We observe that for, $\alfwild{c}$ is an odd function and hence \cref{even:odd:hermite} implies that 
    $\fwild{c}{\tau}\cdot H_{2k}(\tau)$ is an odd function and $\fwild{c}{\tau}\cdot H_{2k+1}(\tau)$ is an 
    even function. This implies for any $k\geq 0$, that $\hfwildc{c}{2k} = 0$ and 
    \[
        \hfwildc{c}{2k+1}
        = 
        \frac{1}{\sqrt{2\pi}} \int_{-\infty}^{\infty} 
        \sgn{(\tau)}\cdot \tau^{c} \cdot H_{2k+1}(\tau) \cdot \ee^{-\tau^2/2}  \,d\tau
        = 
        \sqrt{\frac{2}{\pi}} \int_{0}^{\infty} 
        \tau^{c} \cdot H_{2k+1}(\tau) \cdot \ee^{-\tau^2/2}  \,d\tau \mper
    \]
    Thus we have 
    \begin{align*}
        &\quad~\sum_{k\geq 0} \frac{\lambda^{2k+1}}{\sqrt{(2k+1)!}}\cdot \hfwildc{c}{2k+1} \\
        &= 
        \sqrt{\frac{2}{\pi}} \cdot \sum_{k\geq 0}~\int_{0}^{\infty} 
        \tau^{c} \cdot \ee^{-\tau^2/2} \cdot H_{2k+1}(\tau) \cdot \frac{\lambda^{2k+1}}{\sqrt{(2k+1)!}}\,\,d\tau  \\
        &= 
        \sqrt{\frac{2}{\pi}} \cdot \int_{0}^{\infty} \tau^{c} \cdot \ee^{-\tau^2/2} \sum_{k\geq 0} 
        H_{2k+1}(\tau) \cdot \frac{\lambda^{2k+1}}{\sqrt{(2k+1)!}}\,\,d\tau 
        &&(\text{see below}) \\
        &= 
        \frac{1}{\sqrt{2\pi}} \cdot \int_{0}^{\infty} \tau^{c} \cdot \ee^{-\tau^2/2} \cdot 
        \inparen{\ee^{\tau\lambda  -\lambda^2/2} - \ee^{-\tau\lambda  -\lambda^2/2}} \,\,d\tau
        &&(\text{~by \cref{hermite:generating:function}}) \\
        &= 
        \frac{1}{\sqrt{2\pi}} \cdot \int_{0}^{\infty} \tau^{c} \cdot 
        \inparen{\ee^{-(\tau-\lambda)^2/2} - \ee^{-(\tau+\lambda)^2/2}} \,d\tau 
    \end{align*}
    where the exchange of summation and integral in the second equality follows by Fubini's theorem. 
    We include this routine verification for the sake of completeness.
    As a consequence of Fubini's theorem, if $(f_k:\R\to\R)_k$ is a sequence 
    of functions such that $\sum_{k\geq 0}\int_{0}^{\infty} |f_k| <\infty$, then 
    $
        \sum_{k\geq 0}\int_{0}^{\infty} f_k = \int_{0}^{\infty} \sum_{k\geq 0} f_k \mper
    $
    Now for any fixed $k$, we have 
    \[
        \int_{0}^{\infty} \tau^{c}\cdot |H_{k}(x)|\,d\gamma(\tau) 
        ~\leq ~
        \inparen{\int_{0}^{\infty} \tau^{2c}\,d\gamma(\tau)}^{1/2} \cdot 
        \inparen{\int_{0}^{\infty} |H_{k}(x)|^2\,d\gamma(\tau)}^{1/2} 
        ~\leq ~ 
        \gamma_{2c}^{c}
        < 
        \infty \mper
    \]
    Setting $f_k(\tau) := \tau^{c} \cdot \ee^{-\tau^2/2} \cdot H_{2k+1}(\tau) \cdot 
    \lambda^{2k+1}/\sqrt{(2k+1)!}\,\mcom$ we get that $\sum_{k\geq 0}\int_{0}^{\infty} |f_k| <\infty$. This 
    completes the proof. 
\end{proof}

Finally using known results about parabolic cylinder functions, we are able to relate the aforementioned 
integral representation to a confluent hypergeometric function (whose Taylor coefficients are known). 
\begin{lemma}
\label[lemma]{integral:representation:confHG}
    For $\lambda \in [-1,1]$ and real valued $c>-1$, we have 
    \[
        \frac{1}{\sqrt{2\pi}}
        \int_{0}^{\infty} \tau^{c}\inparen{\ee^{-(\tau-\lambda)^2/2} - \ee^{-(\tau+\lambda)^2/2}}\,d\tau
        ~=~
        \gamma_{c+1}^{c+1}\cdot 
        \lambda\cdot \confHG\inparen{\frac{1-c}{2}\,;\,\frac{3}{2}\,;\,-\frac{\lambda^2}{2}} 
    \]
\end{lemma}

\begin{proof}
    We prove this by using the \cref{parabolic-conf-hypergeometric} with $a=c+\frac{1}{2}$. We note that
    $\alpha>-\frac{1}{2}$ and $\confHG\inparen{\cdot,\cdot,-\lambda^2/2}$ is an even
    function of $\lambda$. So combining the two, we get:
    \begin{align*}
        &\quad~\frac{1}{\sqrt{2\pi}}
        \int_{0}^{\infty} \tau^{c}\inparen{\ee^{-(\tau-\lambda)^2/2} - \ee^{-(\tau+\lambda)^2/2}}\,d\tau \\
        &=~\frac{2}{\sqrt{2\pi}}\cdot\frac{\sqrt{\pi}\cdot\Gamma\inparen{c+1}}{2^{c/2}
        \cdot\Gamma\inparen{\frac{c+1}{2}}}
        \cdot \lambda\cdot\confHG\inparen{-\frac{c}{2}+\frac{1}{2}\,;\frac{3}{2}\,;-\frac{1}{2}\lambda^2}\\
        &=~2^{(1-c)/2}\cdot
        \frac{\Gamma\inparen{\frac{c+1}{2}+\frac{1}{2}}}{2^{-c}\cdot\sqrt{\pi}}\cdot 
        \lambda\cdot \confHG\inparen{\frac{1-c}{2}\,;\,\frac{3}{2}\,;\,-\frac{\lambda^2}{2}} 
        &&(\text{by \cref{duplication:formula}}) \\
        &=~\gamma_{c+1}^{c+1}\cdot 
        \lambda\cdot \confHG\inparen{\frac{1-c}{2}\,;\,\frac{3}{2}\,;\,-\frac{\lambda^2}{2}} 
        &&(\text{by \cref{gamma:and:Gamma}}) \qedforce
    \end{align*}
\let\qed\relax
\end{proof}
\subsection[Taylor Coefficients of f]{Taylor Coefficients of $\fplain{a}{b}{x}$ and Hypergeometric
  Representation}
By \cref{noise:correlation}, we are left with understanding the function whose power series is given by 
a weighted coefficient-wise product of a certain pair of confluent hypergeometric functions. This 
turns out to be precisely the Gaussian hypergeometric function, as we will see below. 
\begin{observation}
\label[observation]{1F1:circ:1F1:gives:2F1}
    Let $f_k := [\tau^k]\, \confHG(a_1,3/2,\tau)$ and $h_k := [\tau^k] \,\confHG(b_1,3/2,\tau)$. Further let \\
    $\mu_k := f_k\cdot h_k\cdot (2k+1)!/4^k$. Then for $\rho\in [-1,1]$, 
    \[
        \sum_{k\geq 0} \mu_k\cdot \rho^n  ~=~  \hypergeometric(a_1,b_1\,;\,3/2\,;\,\rho) \mper
    \]
\end{observation}

\begin{proof}
    The claim is equivalent to showing that
    $\mu_k = \RisingFactorial{a_1}\,\RisingFactorial{b_1}/(\RisingFactorial{3/2}\,k!)$. 
    Since we have $f_k = \RisingFactorial{a_1}/(\RisingFactorial{3/2}\,k!)$ and 
    $h_k = \RisingFactorial{b_1}/(\RisingFactorial{3/2}\,k!)$, it is sufficient to show that 
    $(2k+1)!/4^k = \RisingFactorial{3/2}\cdot k!$. Indeed we have, 
    \begin{align*}
        (2k+1)! 
        &= 
        2^k\cdot k!\cdot 1\cdot 3 \cdot 5 \cdots (2k+1) \\
        &= 
        4^k\cdot k!\cdot \frac{3}{2} \cdot \frac{5}{2} \cdots \inparen{\frac{3}{2}+k-1} \\
        &= 
        4^k\cdot k!\cdot \RisingFactorial{3/2} \mper \qedforce
    \end{align*}
\let\qed\relax
\end{proof}

We are finally equipped to put everything together. 
\begin{theorem}
    \label[theorem]{hypergeometric:representation}
    For any $a,b\in (-1,\infty)$ and $\rho\in [-1,1]$, we have
    \[
        \nf{a}{b}{\rho}
        ~:=~
        \frac{1}{\gamma_{a+1}^{a+1}\cdot \gamma_{b+1}^{b+1}}\cdot 
        \Ex{\bfg_1 \sim_\rho \,\bfg_2}{\sgn(\bfg_1)|\bfg_1|^{a}\sgn(\bfg_2)|\bfg_1|^{b}}
        ~=~
        \hgp{\rho} \mper
    \]
    It follows that the $(2k+1)$-th Taylor coefficient of $\nf{a}{b}{\rho}$ is 
    \[
        \frac{\RisingFactorial{(1-a)/2}\,\RisingFactorial{(1-b)/2}}{(\RisingFactorial{3/2}\,k!)} \mper
    \]
\end{theorem}

\begin{proof}
    The claim follows by combining \cref{noise:correlation}, 
    \cref{gen:function:integral:rep,integral:representation:confHG}, and \cref{1F1:circ:1F1:gives:2F1}. 
\end{proof}

This hypergeometric representation immediately yields some non-trivial coefficient and 
monotonicity properties: 
\begin{corollary}
\label[corollary]{monotonicity:properties}
    For any $a,b\in [0,1]$, the function $\alnf{a}{b}:[-1,1]\to \R$ satisfies 
    \begin{enumerate}[(M1)]
        \item $[\rho]\,\nf{a}{b}{\rho} = 1$ ~and~ $[\rho^{3}]\,\nf{a}{b}{\rho} = (1-a)(1-b)/6$. 
        
        \item All Taylor coefficients are non-negative. Thus $\nf{a}{b}{\rho}$ 
        is increasing on $[-1,1]$. 
        
        \item All Taylor coefficients are decreasing in $a$ and in $b$. Thus for any fixed $\rho\in [-1,1]$,~ 
        $\nf{a}{b}{\rho}$ is decreasing in $a$ and in $b$. 
        
        \item Note that $\nf{a}{b}{0}=0$ and by (M1) and (M2), $\nf{a}{b}{1} \geq 1$. By continuity, 
        $\nf{a}{b}{[0,1]}$ contains $[0,1]$. Combining this with (M3) implies that for any fixed 
        $\rho\in [0,1]$, $\nfin{a}{b}{\rho}$ is increasing in $a$ and in $b$. 
    \end{enumerate}
\end{corollary}


\section[Bound on Defect]{$\sinh^{-1}(1)/(1+\eps_0)$ ~Bound on $\nhin{a}{b}{1}$}
In this section we show that $p=\infty, q=1$ (the Grothendieck case) is roughly the extremal case for the 
value of $\nhin{a}{b}{1}$, \ie we show that for any $1\leq q\leq 2 \leq p \leq \infty$, $\nhin{a}{b}{1}\geq 
\sinh^{-1}(1)/(1+\eps_0)$ (recall that $\nhin{0}{0}{1} = \sinh^{-1}(1)$). While we were unable to establish as 
much, we conjecture that $\nhin{a}{b}{1} \geq \sinh^{-1}(1)$. 
\cref{behavior:of:coefficients} details some of the challenges involved in establishing that $\sinh^{-1}(1)$ is the worst case, and presents our approach to establish an approximate bound, which will be formally 
proved in \cref{bounding:coefficients}.

\subsection[Behavior of The Coefficients of The Inverse Function]
{Behavior of The Coefficients of $\nfin{a}{b}{z}$.}
\label[subsection]{behavior:of:coefficients}
Krivine's upper bound on the real Grothendieck constant, 
Haagerup's upper bound \cite{Haagerup81} on the complex Grothendieck constant and the work of 
Naor and Regev~\cite{NR14, BFV14} on the optimality of Krivine schemes are all closely related to our work in 
that each of the aforementioned papers needs to lower bound $(\absolute{f^{-1}})^{-1}(1)$ for an appropriate 
odd function $f$ (the work of Briet \etal~\cite{BFV14} on the rank-constrained Grothendieck problem is 
also a generalization of Krivine's and Haagerup's work, however they did not derive a closed form upper 
bound on $(\absolute{f^{-1}})^{-1}(1)$ in their setting). In Krivine's setting $f=\sin^{-1} x$, implying 
$(\absolute{f^{-1}})^{-1} = \sinh^{-1}$ and hence the bound is immediate. In our setting, as well as 
in \cite{Haagerup81} and \cite{NR14, BFV14}, $f$ is  given by its Taylor coefficients and is not known to have 
a closed form. In \cite{NR14}, all coefficients of $f^{-1}$ subsequent to the third are negligible and so one 
doesn't incur much loss by assuming that $\absolute{f^{-1}}(\rho) = c_1\rho + c_3\rho^3$.
In \cite{Haagerup81}, the coefficient of $\rho$ in $f^{-1}(\rho)$ is $1$ and every subsequent coefficient is 
negative, which implies that $\absolute{f^{-1}}(\rho) = 2\rho - f^{-1}(\rho)$. 
Note that if the odd coefficients of $f^{-1}$ are alternating in sign like in Krivine's setting, then 
$\absolute{f^{-1}}(\rho) = -i \cdot f^{-1}(i\rho)$. 
These structural properties of the coefficients help their analyses. 

In our setting there does not appear to be such a strong relation between $(\absolute{f^{-1}})$ and $f^{-1}$. 
Consider $f(\rho) = \nf{a}{a}{\rho}$. For certain $a\in (0,1)$, the sign pattern of the coefficients of 
$f^{-1}$ is unlike that of \cite{Haagerup81} or $\sin \rho$. In fact empirical results suggest that the odd 
coefficients of $f$ alternate in sign up to some term $K=K(a)$, and subsequently the coefficients 
are all non-positive (where $K(a) \to \infty$ as $a\to 0$), \ie the sign pattern appears to be interpolating 
between that of $\sin \rho$ and that of $f^{-1}(\rho)$ in the case of Haagerup~\cite{Haagerup81}.
 
Another source of difficulty is that for a fixed $a$, the coefficients of $f^{-1}$ (with and without magnitude) 
are not necessarily monotone in $k$, and moreover for a fixed $k$, the $k$-th coefficient of $f^{-1}$ is not 
necessarily monotone in $a$. 

A key part of our approach is noting that certain milder assumptions on the coefficients are sufficient to 
show that $\sinh^{-1}(1)$ is the worst case. The proof crucially uses the monotonicity of $\nf{a}{b}{\rho}$ in 
$a$ and $b$. The conditions are as follows: \medskip 

\noindent
Let $\ngc := [\rho^k]\,\nfin{a}{b}{\rho}$. Then 
\begin{enumerate}[(C1)]
    \item $\ngc \leq 1/k!$ ~if~ $k\! \pmod{4} \equiv 1$. 
    
    \item $\ngc \leq 0$ ~if~ $k\! \pmod{4} \equiv 3$.
\end{enumerate}
To be more precise, we were unable to establish that the above conditions hold for all $k$ (however 
we conjecture that it is true for all $k$), and instead use 
Mathematica to verify it for the fist few coefficients. We additionally show that the coefficients of 
$\alnfin{a}{b}$ decay exponentially. Combining this exponential decay with a robust version of the previously 
advertised claim yields that $\nhin{a}{b}{1} \geq \sinh^{-1}(1)/(1 + \eps_0)$. 

We next proceed to prove the claim that the aforementioned conditions are sufficient to show that 
$\sinh^{-1}(1)$ is the worst case. We will need the following definition.
For an odd positive integer $t$, let 
\[
    h_{err}(t,\rho) := \sum_{k\geq t} |\ngc|\cdot \rho^{k}
\]

\begin{lemma}
\label[lemma]{bound:defect}
    If $t$ is an odd integer such that (C1) and (C2) are satisfied for all $k<t$, and 
    $\rho = \sinh^{-1}(1-2 h_{err}(t,\delta))$ ~for some $\delta\geq \rho$,~then $\nh{a}{b}{\rho}\leq 1$. 
\end{lemma}

\begin{proof}
    We have, 
    \begin{align*}
        &\quad~ \nh{a}{b}{\rho}  \\
        &= 
        \sum_{k\geq 1} |\ngc|\cdot \rho^k  \\
        &= 
        -\nfin{a}{b}{\rho} 
        ~+ 
        \sum_{k\geq 1} \max\{2\ngc,0\}\cdot \rho^k  \\
        &= 
        -\nfin{a}{b}{\rho} 
        ~+ 
        \sum_{\substack{1\leq k < t \\ k \text{ mod }4 \equiv 1}} \max\{2\ngc,0\}\cdot \rho^k  
        ~+~
        \sum_{k\geq t} \max\{2\ngc,0\} \cdot \rho^k
        &&(\text{by (C2)}) \\
        &\leq 
        -\nfin{a}{b}{\rho} 
        ~+ 
        \sum_{\substack{1\leq k < t \\ k \text{ mod } 4 \equiv 1}} \max\{2\ngc,0\}\cdot \rho^k  
        ~+~
        2 \, h_{err}(t,\rho) \\
        &\leq 
        -\nfin{a}{b}{\rho} 
        ~+ 
        \sin(\rho) + \sinh(\rho)
        ~+~
        2 \, h_{err}(t,\rho)
        &&(\text{by (C1)}) \\
        &\leq 
        -\nfin{a}{b}{\rho} 
        ~+ 
        \sin(\rho) + 1 + 2(h_{err}(t,\rho) - h_{err}(t,\delta))
        &&(\,\rho = \sinh^{-1}(1-2 h_{err}(t,\delta))\,) \\
        &\leq 
        -\nfin{a}{b}{\rho} 
        ~+ 
        \sin(\rho) + 1 
        &&(\rho\leq \delta) \\
        &\leq 
        -\nfin{0}{0}{\rho} 
        ~+ 
        \sin(\rho) + 1  
        &&(\text{\cref{monotonicity:properties} : (M4)})\\
        &= 
        1 
        &&(\nfin{0}{0}{\rho} = \sin(\rho)) \qedforce
    \end{align*}
\let\qed\relax
\end{proof}

Thus we obtain, 
\begin{theorem}
    For any $1\leq q \leq 2 \leq p \leq \infty$, let $a:=p^*-1,b=q-1$. Then for any $m,n\in \N$ and 
    $A\in \R^{m\times n}$,~~ 
    $\CP{A}/\norm{p}{q}{A} ~\leq~ 1/(\nhin{a}{b}{1}\cdot \gamma_{q}\,\gamma_{p^*})$ ~
    and moreover 
    \begin{itemize}
        \item $\nhin{1}{b}{1} = \nhin{a}{1}{1} = 1$.
        \item $\nhin{a}{b}{1}\geq \sinh^{-1}(1)/(1 + \eps_0)$ ~where $\eps_0 = 0.00863$. 
    \end{itemize}
\end{theorem}

\begin{proof}
    The first inequality follows from \cref{defect:bound:implies:apx}. 
    As for the next item, If $p=2$ or $q=2$ ~(\ie $a=1$ or $b=1$) we are trivially done since 
    $\nhin{a}{b}{\rho} = \rho$ in that case (since for $k\geq 1$,~ $\RisingFactorial{0}=0$). 
    So we may assume that $a,b\in [0,1)$. 
    
    We are left with proving the final part of the claim. 
    Now using Mathematica we verify (exactly)\footnote{
        We generated $\ngc$ as a polynomial in $a$ and $b$ and maximized it over $a,b\in [0,1]$ using 
        the Mathematica ``Maximize'' function which is exact for polynomials.
    }
    that $(C1)$ and $(C2)$ are true for $k\leq 29$. Now let $\delta=\sinh^{-1}(0.974203)$. 
    Then by \cref{inv:coeff:bound} (which states that $\ngc$ decays exponentially and will be proven in the 
    subsequent section), 
    \[
        h_{err}(31,\delta) :=  \sum_{k\geq 31} |\ngc|\cdot d^{k} 
        \leq  \frac{6.1831}{31}\cdot \frac{\delta^{31}}{1-\delta^2} 
        \leq  0.0128991\dots \mper
    \] 
    Now by \cref{bound:defect} we know $\nhin{a}{b}{1} \geq \sinh^{-1}(1-2h_{err}(31,\delta))$. 
    Thus, \\$\nhin{a}{b}{1}\geq \sinh^{-1}(0.974202) \geq \sinh^{-1}(1)/(1+\eps_0)$ for 
    $\eps_0 = 0.00863$, which completes the proof. 
\end{proof}

\subsection{Bounding Inverse Coefficients}
\label[subsection]{bounding:coefficients}
In this section we prove that $\ngc$ decays as $1/c^k$ for some $c=c(a,b)>1$, proving 
\cref{inv:coeff:bound}. 
Throughout this section we assume $1\leq p^*,q <2$, and $a=p^*-1,~b=q-1$ (\ie $a,b\in [0,1)$).
Via the power series representation, $\nf{a}{b}{z}$ can be analytically continued to the unit complex disk. 
Let $\nfin{a}{b}{z}$ be the inverse of $\nf{a}{b}{z}$ and recall $\ngc$ denotes its $k$-th Taylor coefficient. 

We begin by stating a standard identity from complex analysis that provides a convenient contour integral 
representation of the Taylor coefficients of the inverse of a function. We include a proof for completeness. 
\begin{lemma}[Inversion Formula]
\label[lemma]{residue} 
    There exists $\delta>0$, such that for any odd $k$, 
    \begin{equation}
    \label[equation]{inversion:formula}
        \ngc = \frac{2}{\pi k} \,\Im\!\inparen{\int_{C^+_\delta} \nf{a}{b}{z}^{-k}\,dz}
    \end{equation}
    where $C^+_\delta$ denotes the first quadrant quarter circle of radius 
    $\delta$ with counter-clockwise orientation. 
\end{lemma}

\begin{proof}
    Via the power series representation, $\nf{a}{b}{z}$ can be analytically continued to the unit complex disk. 
    Thus by inverse function theorem for holomorphic functions, there exists $\delta_0 \in (0,1]$ such that 
    $\nf{a}{b}{z}$ has an analytic inverse in the open disk $|z|<\delta_0$. So for $\delta \in (0,\delta_0)$, 
    $\nf{a}{b}{C_\delta}$ is a simple closed curve with winding number $1$ (where $C_\delta$ is the complex 
    circle of radius $\delta$ with the usual counter-clockwise orientation). Thus by Cauchy's integral formula 
    we have 
    \[
        \ngc 
        ~=~
        \frac{1}{2\pi i}\,\int_{\nf{a}{b}{C_\delta}} \frac{\nfin{a}{b}{w}}{w^{\,k}}\,dw
        ~~=~ 
        \frac{1}{2\pi i}\,\int_{C_\delta} \,\frac{z \cdot \alnf{a}{b}'(z)}{\nf{a}{b}{z}^{k+1}} \,\,dz
    \]
    where the second equality follows from substituting $w = \nf{a}{b}{z}$. 
    
    Now by \cref{no:roots},~\, $z/\nf{a}{b}{z}^k$~ is holomorphic on the open set $|z|\in (0,1)$, which 
    contains $C_{\delta}$. Hence by the fundamental theorem of contour integration we have 
    \[
        \int_{C_{\delta}} \frac{d}{dz}\inparen{\frac{z}{\nf{a}{b}{z}^{k}}}\,dz 
        ~=~ 
        0 
        \quad~~ \Rightarrow \quad
        \int_{C_{\delta}} \,\frac{z \cdot \alnf{a}{b}'(z)}{\nf{a}{b}{z}^{k+1}} \,\,dz
        ~=~
        \frac{1}{k}\,\int_{C_{\delta}} \,\frac{1}{\nf{a}{b}{z}^{k}}\,\,dz 
    \]
    So we get, 
    \[
        \ngc 
        ~=~
        \frac{1}{2\pi i k}\,\int_{C_\delta} \nf{a}{b}{z}^{-k}\,dz
        ~~=~ 
        \frac{1}{2\pi k}\,
        \Im\!\inparen{
        \int_{C_\delta} \nf{a}{b}{z}^{-k}\,dz
        }
    \]
    where the second equality follows since $\ngc$ is purely real. Lastly, we complete the proof of the 
    claim by using the fact that for odd $k$,~ $\nf{a}{b}{z}^{-k}$ is odd and that 
    $\overline{\nf{a}{b}{z}} = \nf{a}{b}{\overline{z}}$.
\end{proof}

We next state a standard bound on the magnitude of a contour integral that we will use in our analysis. 
\begin{fact}[ML-inequality]
\label[fact]{ML:inequality}
    If $f$ is a complex valued continuous function on a contour $\Gamma$ and $|f(z)|$ is bounded by $M$ for 
    every $z\in \Gamma$, then 
    \[
        \abs{\int_{\Gamma}f(z)} \leq M\cdot \ell(\Gamma)
    \]
    where $\ell(\Gamma)$ is the length of $\Gamma$. 
\end{fact}
Unfortunately the integrand in \cref{inversion:formula} can be very large for small $\delta$, and we cannot 
use the ML-inequality as is. To fix this, we modify the contour of integration (using Cauchy's integral 
theorem) so that the imaginary part of the integral vanishes when restricted to the sections close to 
the origin, and the integrand is small in magnitude on the sections far from the origin (thus allowing us to 
use the ML-inequality). To do this we will need some preliminaries. 

$\nf{a}{b}{z}$ is defined on the closed complex unit disk. The domain is analytically extended to 
the region $\C\setminus ((-\infty,-1) \cup (1,\infty))$, using the Euler-type integral representation of the hypergeometric function. 
\[
    \nfext{a}{b}{z} 
    := 
    \betaterm^{-1}\cdot 
    \intfull{z}
\]
where $\mathrm{B}(\tau_1,\tau_2)$ is the beta function and 
\[
    \intfull{z}
    :=  
    z \int_{0}^{1} \frac{(1-t)^{b/2} \,dt}{t^{(1+b)/2}\cdot (1-z^2 t)^{(1-a)/2}}.
\]

\begin{fact}
\label[fact]{no:roots}
    For any $a_1>0$, $\hypergeometric(a_1,b_1,c_1,z)$ has no non-zero roots in the region 
    $\C\setminus  (1,\infty)$. This implies that if $p^*<2$,~ $\nfext{a}{b}{z}$ has no non-zero roots in 
    the region $\C\setminus ((-\infty,-1) \cup (1,\infty))$. 
\end{fact}

We are now equipped to expand the contour. Our choice of contour is inspired by that of
Haagerup~\cite{Haagerup81} 
which he used in deriving an upper bound on the complex Grothendieck constant. The contour we choose 
has some differences for technical reasons related to the region to which hypergeometric 
functions can be analytically extended. The analysis is quite different from that of Haagerup 
since the functions in consideration behave differently. In fact the inverse function Haagerup considers 
has polynomially decaying coefficients while the class of inverse functions we consider have coefficients 
that have decay between exponential and factorial. 

\begin{observation}[Expanding Contour]
    For any $\alpha \geq 1$ and $\eps>0$, let $P(\alpha,\eps)$ be the four-part curve (see \cref{contour}) 
    given by 
    \begin{itemize}
        \item the line segment $\delta~ \rightarrow ~(1-\eps)$, 
        
        \item the line segment $(1-\eps)~ \rightarrow ~ (\sqrt{\alpha-\eps} + i\sqrt{\eps})$ 
        (henceforth referred to as $L_{\alpha,\eps}$), 
        
        \item the arc along $C^+_\alpha$ starting at $(\sqrt{\alpha-\eps} + i\sqrt{\eps})$ and ending at $i\alpha$ 
        (henceforth referred to as $C^+_{\alpha,\eps}$), 
        
        \item the line segment $i\alpha~ \rightarrow ~i\delta$. 
    \end{itemize}
    By Cauchy's integral theorem, combining \cref{residue} with \cref{no:roots} yields that for odd $k$, 
    \[
        \ngc = \frac{2}{\pi k}\, \Im\!\inparen{\int_{P(\alpha,\eps)} \nfext{a}{b}{z}^{-k}\,dz}
    \]
\end{observation}
\begin{figure}
\caption{The Contour $P(\alpha,\eps)$}
\label[figure]{contour}
\includegraphics{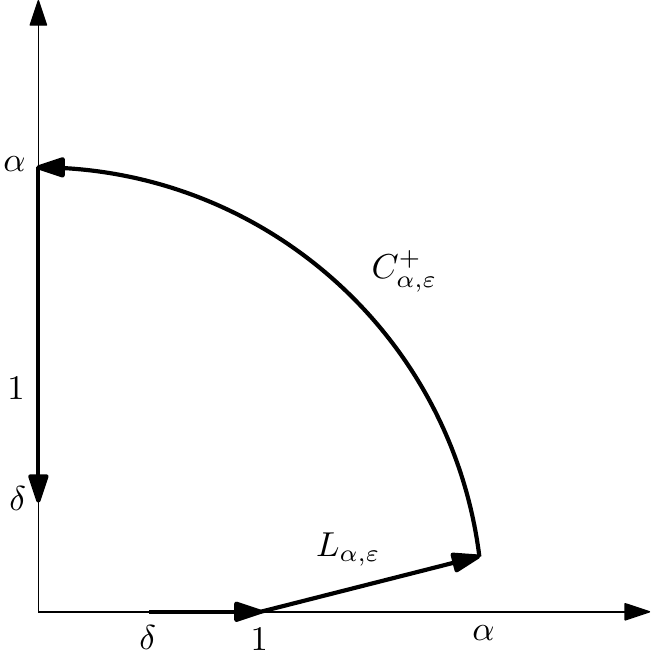}
\centering
\end{figure}

We will next see that the imaginary part of our contour integral vanishes on section of $P(\alpha,\eps)$. 
Applying ML-inequality to the remainder of the contour, combined with lower bounds on 
$|\nfext{a}{b}{z}|$ ~(proved below the fold in \cref{largeness:of:f}), allows us to derive an exponentially 
decaying upper bound on $|\ngc|$. 
\begin{lemma}
\label[lemma]{inv:coeff:bound}
    For any $1\leq p^*,q<2$,  there exists $\eps>0$ such that 
    \[
        |\ngc| \leq \frac{6.1831}{k(1+\eps)^{k}}.
    \]
\end{lemma} 

\begin{proof}
    For a contour $P$, we define $V(P)$ as  
    \[
        V(P) := \frac{2}{\pi k} \,\Im\!\inparen{\int_{P} \nfext{a}{b}{z}^{-k}\,dz}
    \]
    As is evident from the integral representation, $\nfext{a}{b}{z}$ is purely imaginary if $z$ is 
    purely imaginary, and as is evident from the power series, $\nf{a}{b}{z}$ is purely real if $z$ lies on 
    the real interval $[-1,1]$. This implies that $V(\delta \rightarrow (1-\eps)) = 
    V(i\alpha\rightarrow i\delta) = 0$. 
    
    Now combining \cref{ML:inequality} (ML-inequality) with \cref{|f(z)|:incr:in:|z|} and 
    \cref{close:to:real:contour} (which state that the integrand is small in magnitude over $C^+_{6,\eps}$ 
    and $L_{6,\eps}$ respectively), we get that for sufficiently small $\eps>0$, 
    \begin{align*}
        |V(P(6,\eps))| 
        &\leq 
        |V(C^+_{6,\eps})| + |V(L_{6,\eps})|  \\
        &\leq 
        \frac{2}{\pi k}\cdot \frac{3\pi/2}{(1+\eps)^{k}} + 
        \frac{2}{\pi k}\cdot \frac{6-1+O(\sqrt{\eps})}{(1+\eps)^k} \\
        &\leq 
        \frac{6.1831}{k (1+\eps)^{k}}. 
        &&(\text{taking } \eps \text{ sufficiently small}) \qedforce
    \end{align*}
\let\qed\relax
\end{proof}

\subsubsection{Lower bounds on $|\nfext{a}{b}{z}|$ Over $C^+_{\alpha,\eps}$ and $L_{\alpha,\eps}$}
\label[section]{largeness:of:f}
In this section we show that for sufficiently small $\eps$, $|\nfext{a}{b}{z}|>1$ over $L_{\alpha,\eps}$ (regardless of the value of $\alpha$, Lemma~\ref{close:to:real:contour}), and over $C^+_{\alpha,\eps}$ when $\alpha$ is a sufficiently large constant (Lemma~\ref{|f(z)|:incr:in:|z|}). 

We will first show the claim for $C^+_{\alpha,\eps}$ by relating $|\nfext{a}{b}{z}|$ to $|z|$. While the 
asymptotic behavior of hypergeometric functions for $|z|\to \infty$ has been extensively studied (see 
for instance \cite{Lozier03}), it appears that our desired estimates aren't immediate consequences of 
prior work for two reasons. Firstly, we require relatively precise estimates for moderately large but 
constant $|z|$. Secondly, due to the expressive power of hypergeometric functions, the estimates we 
derive can only be true for hypergeometric functions parameterized in a specific range. Indeed, our proof 
crucially uses the fact that $a,b\in [0,1)$. Our approach is to use the Euler-type integral representation 
of $\nfext{a}{b}{z}$ which as a reminder to the reader is as follows: 
\[
    \nfext{a}{b}{z} 
    := 
    \betaterm^{-1}\cdot 
    \intfull{z}
\]
where $\mathrm{B}(x,y)$ is the beta function and 
\[
    \intfull{z}
    :=  
    z \int_{0}^{1} \frac{(1-t)^{b/2} \,dt}{t^{(1+b)/2}\cdot (1-z^2 t)^{(1-a)/2}}.
\]
We start by making the simple observation that the integrand of $\intfull{z}$ is always in the positive 
complex quadrant --- an observation that will come in handy multiple times in this section, in dismissing 
the possibility of cancellations. This is the part of our proof that makes the most crucial use of 
the assumption that $0\leq a<1$ (equivalently $1\leq p^* <2$). 
\begin{observation}
\label[observation]{reim:monotonicity}
    Let $z = r e^{i\theta}$ be such that either one of the following two cases is satisfied:
    \begin{enumerate}[(A)]
        \item $r<1$ and $\theta = 0$.
        
        \item $\theta\in (0,\pi/2]$.
    \end{enumerate}     
    Then for any $0 \leq a\leq 1$ and any $t\in \R^{+}$, 
    \[
        \arg\inparen{\frac{z}{(1 - tz^2)^{(1-a)/2}}} \in [0,\pi/2]
    \]
\end{observation}

\begin{proof}
    The claim is clearly true when $\theta = 0$ and $r<1$. It is also clearly true when $\theta=\pi/2$. 
    Thus we may assume $\theta\in (0,\pi/2)$.  
    \begin{align*}
        & \arg(z) \in (0,\pi/2) 
        \quad\Rightarrow 
        \arg(-tz^2) \in (-\pi,0) 
        \quad\Rightarrow 
        \Im(-tz^2) < 0 \\
        &\Rightarrow 
        \Im(1-tz^2) < 0
        \quad\Rightarrow 
        \arg(1 - tz^2) \in (-\pi,0) 
    \end{align*}
    Moreover since $\arg(-tz^2)=2\theta-\pi \in (-\pi,0)$, we have $\arg(1-tz^2)> 2\theta-\pi$. 
    Thus we have, 
    \begin{align*}
        & \arg(1-tz^2) \in (2\theta - \pi,0) 
        \quad\Rightarrow 
        \arg\inparen{(1-tz^2)^{(1-a)/2}} \in ((1-a)(\theta-\pi/2),0) \\
        &\Rightarrow 
        \arg\inparen{1/(1-tz^2)^{(1-a)/2}} \in (0,(1-a)(\pi/2-\theta))  \\
        &\Rightarrow 
        \arg\inparen{z/(1-tz^2)^{(1-a)/2}} \in (0,(1-a)(\pi/2-\theta) +\theta) \subseteq (0,\pi/2)
\qedforce
    \end{align*}
\let\qed\relax
\end{proof}

We now show $|\nfext{a}{b}{z}|$ is large over $C^+_{\alpha,\eps}$. The main idea is to move from a 
complex integral to a real integral with little loss, and then estimate the real integral. To do this, we use 
\cref{reim:monotonicity} to argue that the magnitude of $\intfull{z}$ is within $\sqrt{2}$ of the 
integral of the magnitude of the integrand. 

\begin{lemma}[ \bf{$|\nfext{a}{b}{z}|$ is large over $C^+_{\alpha,\eps}$} ] 
\label[lemma]{|f(z)|:incr:in:|z|}
    Assume $a,b\in[0,1)$ and consider any $z\in \C$ with $|z| \geq 6$. Then $|\nfext{a}{b}{z}| > 1$. 
\end{lemma}
\begin{proof}
    We start with a useful substitution. 
    \begin{align*}
        \intfull{z}
        &=  z \int_{0}^{1} \frac{(1-t)^{b/2} \,dt}{t^{(1+b)/2}\cdot (1-z^2 t)^{(1-a)/2}} \\
        &= 
         r^{b}e^{i\theta} 
         \int_{0}^{r^2} \frac{(1-s/r^2)^{b/2} \,ds}{s^{(1+b)/2}\cdot (1-e^{2i\theta} s)^{(1-a)/2}} 
         && (\text{Subst. } s = r^2 t, \text{ where } z=re^{i\theta}) \\
         &= 
         r^{b}
         \int_{0}^{r^2} \frac{w_a(s,\theta)\cdot (1-s/r^2)^{b/2} \,ds}{s^{1+(b-a)/2}} 
    \end{align*}
    where
    \[
        w_a(s,\theta) 
        := 
        \frac{e^{i\theta}}{(1/s-e^{2i\theta})^{(1-a)/2}} \mper
    \]
    
    \medskip 
    
    We next exploit the observation that the integrand is always in the positive complex quadrant by 
    showing that $|\intfull{z}|$ is at most a factor of $\sqrt{2}$ away from the integral obtained 
    by replacing the integrand with its magnitude. 
    \begin{align*}
        &\quad~~ |\intfull{z}|  \\
        &=~ 
        \sqrt{\Re(\intfull{z})^2 + \Im(\intfull{z})^2} \\
        &\geq~ 
        (|\Re(\intfull{z})| + |\Im(\intfull{z})|)/\sqrt{2} &&(\text{Cauchy-Schwarz}) \\
        &=~ 
        (\Re(\intfull{z}) + \Im(\intfull{z}))/\sqrt{2} &&(\text{by \cref{reim:monotonicity}}) \\
        &=~
        \frac{r^{b}}{\sqrt{2}}
        \int_{0}^{r^2} 
        \inparen{ \Re(w_a(s,\theta)) + \Im(w_a(s,\theta)) } \cdot 
        \frac{(1-s/r^2)^{b/2} \,ds}{s^{1+(b-a)/2}} \\
        &=~
        \frac{r^{b}}{\sqrt{2}}
        \int_{0}^{r^2} 
        \inparen{ |\Re(w_a(s,\theta))| + |\Im(w_a(s,\theta))| } \cdot 
        \frac{(1-s/r^2)^{b/2} \,ds}{s^{1+(b-a)/2}} 
        &&(\text{by \cref{reim:monotonicity}})\\
        &\geq~
        \frac{r^{b}}{\sqrt{2}}
        \int_{0}^{r^2} 
        |w_a(s,\theta)| \cdot 
        \frac{(1-s/r^2)^{b/2} \,ds}{s^{1+(b-a)/2}} 
        &&(\norm{1}{v}\geq \norm{2}{v}) \\
        &\geq~
        \frac{r^{b}}{\sqrt{2}}
        \int_{0}^{r^2} 
        \frac{1}{(1+1/s)^{(1-a)/2}} \cdot 
        \frac{(1-s/r^2)^{b/2} \,ds}{s^{1+(b-a)/2}} 
    \end{align*}
    We now break the integral into two parts and analyze them separately. We start by analyzing the part that's  
    large when $b\rightarrow 0$.  
    \begin{align*}
        &\quad~~
        \frac{r^{b}}{\sqrt{2}}
        \int_{1}^{r^2} 
        \frac{1}{(1+1/s)^{(1-a)/2}} \cdot 
        \frac{(1-s/r^2)^{b/2} \,ds}{s^{1+(b-a)/2}} \\
        &\geq~
        \frac{r^{b}}{2}
        \int_{1}^{r^2} 
        \frac{(1-s/r^2)^{b/2} \,ds}{s^{1+(b-a)/2}} \\
         &\geq~ 
         \frac{r^b}{2} \int_{1}^{r^2/2} 
         \frac{(1-s/r^2)^{b/2} \,ds}{s^{1+(b-a)/2}} \\
         &\geq~ 
         \frac{r^b}{2\sqrt{2}} \int_{1}^{r^2/2} 
         \frac{ds}{s^{1+(b-a)/2}} 
         &&(\text{since } s\leq r^2/2) \\
         &\geq~ 
         \frac{r^b\cdot \min\{1,r^{a-b}\}}{2\sqrt{2}} \int_{1}^{r^2/2} 
         \frac{ds}{s} \\
         &=~ 
         \frac{\min\{r^a,r^b\}\cdot \log (r^2/2)}{2\sqrt{2}} \\
         &\geq~ 
         \frac{\log (r/\sqrt{2})}{\sqrt{2}} 
    \end{align*}
    We now analyze the part that's large when $b\rightarrow 1$.
    \begin{align*}
        &\quad~~
        \frac{r^{b}}{\sqrt{2}}
        \int_{0}^{1} 
        \frac{1}{(1+1/s)^{(1-a)/2}} \cdot 
        \frac{(1-s/r^2)^{b/2} \,ds}{s^{1+(b-a)/2}} \\
        &=~
        \frac{r^{b}}{\sqrt{2}}
        \int_{0}^{1} 
        \frac{(1-s/r^2)^{b/2} }{(1+s)^{(1-a)/2}} \cdot 
        \frac{ds}{s^{(1+b)/2}} \\
        &\geq~
        \frac{r^{b}\cdot \sqrt{1-1/r^2}}{2}
        \int_{0}^{1} 
        \frac{ds}{s^{(1+b)/2}} 
        &&(\text{since } s\leq 1) \\
        &=~
        \frac{r^{b}\cdot \sqrt{1-1/r^2}}{1-b} 
    \end{align*}
    Combining the two estimates above yields that if $r>\sqrt{2}$, 
    \[
        |\nfext{a}{b}{z}|
        \geq~
        \betaterm^{-1} \cdot 
        \inparen{
        \frac{\log (r/\sqrt{2})}{\sqrt{2}} 
        + 
        \frac{r^{b}\cdot \sqrt{1-1/r^2}}{1-b} 
        }
    \]
    Lastly, the proof follows by using the following estimate:
    \begin{fact}
    \label[fact]{mathematica:fact:beta}
        Via Mathematica, for $0\leq b< 1$ we have 
        \[
            \betaterm^{-1} \cdot 
            \inparen{
            \frac{\log (6/\sqrt{2})}{\sqrt{2}} 
            + 
            \frac{6^{b}\cdot \sqrt{1-1/6^2}}{1-b} 
            }
            ~~\geq~~
            1.003
        \]
    \end{fact}
\end{proof}

\begin{remark}
    The preceding proof can be used to derive the precise asymptotic behavior of 
    $|\nfext{a}{b}{z}|$ in $r$. Specifically, it grows as ~$r^{a} \log r$~ if ~$a=b$~ and as ~$r^{\,\max\{a,b\}}$ 
    ~if ~$a\neq b$. 
\end{remark}

\bigskip

We now show that $|\nfext{a}{b}{z}|>1$ over $L_{\alpha, \eps}$. To do this, it is insufficient to assume 
that $|z|\geq 1$ since there exist points $z$ (for instance $z=i$) of unit length such that 
$|\nfext{a}{b}{z}|<1$. To show the claim, we observe that $|\nfext{a}{b}{z}|$ is large when $z$ is close to 
the real line and use the fact that $L_{\alpha,\eps}$ is close to the real line. Formally, we show that if 
$z$ is of length at least $1$ and is sufficiently close to the real line, $|\nfext{a}{b}{z}|$ is close to 
$\nfext{a}{b}{1}$. Lastly, we use the power series representation of the hypergeometric function to obtain 
a sufficiently accurate lower bound on $\nfext{a}{b}{1}$. 

\begin{lemma}[ \bf{$|\nfext{a}{b}{z}|$ is large over $L_{\alpha,\eps}$} ] 
\label[lemma]{close:to:real:contour}
    Assume $a,b\in [0,1)$ and consider any $\gamma\geq 1-\eps_1$. 
    Let $\eps_2 := \sqrt{\eps_1}$ and $z := \gamma (1+i\eps_1)$. 
    Then for $\eps_1>0$ sufficiently small, $|\nfext{a}{b}{z}|>1$. 
\end{lemma}

\begin{proof}
    Below the fold we will show 
    \begin{align}
        |\intfull{z}| 
        ~\geq~ 
        (1-O(\sqrt{\eps_1}))\int_{0}^{1-\eps_2} \frac{(1-s)^{b/2} \,ds}{s^{(1+b)/2}\cdot (1 - s)^{(1-a)/2}} & 
        \label{ineq:real:closeness} 
    \end{align}
    But we know ($\mathrm{LHS},~\mathrm{RHS}$ refer to \cref{ineq:real:closeness})
    \begin{align*}
        \betaterm^{-1}\cdot \mathrm{LHS} &~=~ \nfext{a}{b}{z} \text{~~~and ~~} \\
        \betaterm^{-1}\cdot \mathrm{RHS} ~&\rightarrow~ \nf{a}{b}{1} \text{  ~as~  } \eps_1\rightarrow 0 
    \end{align*}
    Also by \cref{monotonicity:properties} : (M1), (M2),~ $\nf{a}{b}{1} \geq 1+(1-a)(1-b)/6>1$. 
    Thus for $\eps_1$ sufficiently small, we must have $|\nfext{a}{b}{z}|>1$. \medskip

    We now show \cref{ineq:real:closeness}, by comparing integrands point-wise. To do this, we will 
    assume the following closeness estimate that we will prove below the fold: 
    \begin{equation}
    \label[equation]{eq:real:close:final}
        \Re\inparen{\frac{1+i\eps_1}{(1-s(1+i\eps_1)^2)^{(1-a)/2}}}
        = 
        \frac{1-O(\eps_2)}{(1-s)^{(1-a)/2}} \mper
    \end{equation}
    We will also need the following inequality. Since $\gamma\geq 1-\eps_1 = 1-\eps_2^2$, 
    for any $0\leq s \leq 1-\eps_2$, we have 
    \begin{equation}
    \label[equation]{eq:numerator}
        (1-s/\gamma^2)^{b/2} \geq (1-O(\eps_2))\cdot (1-s)^{b/2}.
    \end{equation}   
    Given, these estimates, we can complete the proof of \cref{ineq:real:closeness} as follows: 
    \begin{align*}
        &\quad~\Re(\intfull{z}) \\
        &= 
        \Re\inparen{
         z\int_{0}^{1} \frac{(1-t)^{b/2} \,dt}{t^{(1+b)/2}\cdot (1 - tz^2)^{(1-a)/2}} 
         } \\
         &= 
         \Re\inparen{
         \gamma^{b}(1+i\eps_1)
         \int_{0}^{\gamma^2} \frac{(1-s/\gamma^2)^{b/2} \,ds}{s^{(1+b)/2}\cdot (1-s(1+i\eps_1)^2)^{(1-a)/2}} 
         }
         &&(\text{subst. } s\leftarrow \gamma^2 t)  \\
         &\geq 
         \Re\inparen{
         \gamma^{b}(1+i\eps_1)
         \int_{0}^{1-\eps_2} \frac{(1-s/\gamma^2)^{b/2} \,ds}{s^{(1+b)/2}\cdot (1-s(1+i\eps_1)^2)^{(1-a)/2}} 
         }
         &&(\text{by \cref{reim:monotonicity}}) \\
         &= 
         \gamma^{b}
         \int_{0}^{1-\eps_2} 
         \Re\inparen{ \frac{1+i\eps_1}{(1-s(1+i\eps_1)^2)^{(1-a)/2}} }
         \frac{(1-s/\gamma^2)^{b/2} \,ds}{s^{(1+b)/2}} \\ 
         &\geq 
         (1-O(\eps_2))\cdot \gamma^b
         \int_{0}^{1-\eps_2} \frac{(1-s/\gamma^2)^{b/2} \,ds}{s^{(1+b)/2}\cdot (1 - s)^{(1-a)/2}} 
         &&(\text{by \cref{eq:real:close:final}}) \\
         &\geq 
         (1-O(\eps_2))
         \int_{0}^{1-\eps_2} \frac{(1-s)^{b/2} \,ds}{s^{(1+b)/2}\cdot (1 - s)^{(1-a)/2}} 
         &&(\text{by \cref{eq:numerator}, } \gamma\geq 1-\eps_1) 
    \end{align*}
    
    It remains to establish \cref{eq:real:close:final}, which we will do by considering the numerator 
    and reciprocal of the denominator separately and subsequently using the fact that
    $\Re(z_1 z_2) = \Re(z_1)\Re(z_2) - \Im(z_1)\Im(z_2)$. In doing this, we need to show that the 
    respective real parts are large and respective imaginary parts are small for which the following 
    simple facts will come in handy. 
    \begin{fact}
        \label[fact]{real:part:under:power}
        Let $z = re^{i\theta}$ be such that $\Re{z}\geq 0$ ~(i.e. $-\pi/2\leq \theta \leq \pi/2$). 
        Then for any $0\leq \alpha\leq 1$, 
        \[
            \Re(1/z^{\alpha}) = \cos(-\alpha\theta)/r^{\alpha} = \cos(\alpha\theta)/r^{\alpha} 
            \geq \cos(\theta)/r^{\alpha} = \Re(z)/r^{1+\alpha}
        \]
    \end{fact}

    \begin{fact}
    \label[fact]{imaginary:part:under:power}
        Let $z = re^{-i\theta}$ be such that $\Re{z}\geq 0, \Im{z}\leq 0$ ~(i.e. $0\leq \theta \leq \pi/2$). 
        Then for any $0\leq \alpha\leq 1$, 
        \[
            \Im(1/z^{\alpha}) = \sin(\alpha\theta)/r^{\alpha} 
            \leq \sin(\theta)/r^{\alpha} = -\Im(z)/r^{1+\alpha}
        \]
    \end{fact}
    We are now ready to prove the claimed properties of the reciprocal of the denominator from 
    \cref{eq:real:close:final}. For any $0\leq s\leq 1-\eps_2$ we have, 
    \begin{align}
    \label[equation]{eq:real:close}
        &\quad~\Re\inparen{\frac{1}{(1-s(1+i\eps_1)^2)^{(1-a)/2}}} \nonumber \\
        &=
        \Re\inparen{\frac{1}{(1-s+s\eps_1^2-2is\eps_1)^{(1-a)/2}}} \nonumber \\
        &=
        \frac{1}{(1-s)^{(1-a)/2}}\cdot 
        \Re\inparen{\frac{1}{(1+s\eps_1^2/(1-s)-2i\eps_1/(1-s))^{(1-a)/2}}} \nonumber \\
        &\geq 
        \frac{1}{(1-s)^{(1-a)/2} \cdot (1+O(\eps_1^2/\eps_2^2))^{(3-a)/4}} 
        &&\hspace{-20 pt}(\text{by \cref{real:part:under:power}, and } 1-s\geq \eps_2) \nonumber \\
        &=
        \frac{1-O(\eps_1^2/\eps_2^2)}{(1-s)^{(1-a)/2}} \nonumber \\
        &= 
        \frac{1-O(\eps_2)}{(1-s)^{(1-a)/2}} 
    \end{align}
    Similarly, 
    \begin{align}
    \label[equation]{eq:im:small}
        &\Im\inparen{\frac{1}{(1-s+s\eps_1^2-2is\eps_1)^{(1-a)/2}}} \nonumber \\
        &=
        \frac{1}{(1-s)^{(1-a)/2}}\cdot 
        \Im\inparen{\frac{1}{(1+s\eps_1^2/(1-s)-2i\eps_1/(1-s))^{(1-a)/2}}} \nonumber \\
        &\leq 
        \frac{2\eps_1}{(1-s)^{(1-a)/2}}
        &&(\text{by \cref{imaginary:part:under:power}})
    \end{align}
    Combining \cref{eq:real:close} and \cref{eq:im:small} with the fact that 
    $\Re(z_1 z_2) = \Re(z_1)\Re(z_2) - \Im(z_1)\Im(z_2)$ yields, 
    \[
        \Re\inparen{\frac{1+i\eps_1}{(1-s(1+i\eps_1)^2)^{(1-a)/2}}}
        = 
        \frac{1-O(\eps_2)}{(1-s)^{(1-a)/2}} \mper
    \]
    This completes the proof. 
\end{proof}

\subsubsection[Challenges of Proving (C1) and (C2) for all k]{Challenges of Proving (C1) and (C2) for all $k$}
\label[subsubsection]{coefficient:challenges}
For certain values of $a$ and $b$, the inequalities in (C1) and (C2) leave very little room for error. 
In particular, when $a=b=0$, (C1) holds at equality and (C2) has $1/k!$ additive slack. 
In this special case, it would mean that one cannot analyze the contour integral (for the 
$k$-th coefficient of $\nfin{a}{b}{\rho}$) by using ML-inequality on any section of the contour that is within 
a distance of $\mathrm{exp}(k)$ from the origin. Analytic approaches would require extremely precise 
estimates on the value of the contour integral on parts close to the origin. Other challenges to naive 
approaches come from the lack of monotonicity properties for $\ngc$ (both in $k$ and in $a,b$ - see 
\cref{behavior:of:coefficients})

\appendix
\section{Factorization of Linear Operators}
\label[section]{factorization}
Let $X,Y,E$ be Banach spaces and let $A:X\to Y$ be a continuous linear operator. We say that $A$ 
\emph{factorizes} through $E$ if there exist continuous operators $C:X\to E$ and $B:E\to Y$ such that 
$A=BC$. 
Factorization theory has been a major topic of study in functional analysis, going as far back as 
Grothendieck's famous ``Resume'' \cite{Grothendieck56}. It has many striking applications, like 
the isomorphic characterization of Hilbert spaces and $L_p$ spaces due to 
\Kwapien~\cite{Kwapien72a, Kwapien72b}, connections to type and cotype through the work of 
\Kwapien~\cite{Kwapien72a}, Rosenthal~\cite{Rosenthal73}, Maurey~\cite{Maurey74} and 
Pisier~\cite{Pisier80}, connections to Sidon sets through the work of Pisier~\cite{Pisier86}, 
characterization of weakly compact operators due to Davis \etal~\cite{DFJP74}, connections to the 
theory of $p$-summing operators through the work of Grothendieck~\cite{Grothendieck56}, 
Pietsch~\cite{Pietsch67} and Lindenstrauss and Pelczynski~\cite{LP68}. 

Let $\factorConst{A}$ denote 
\[
    \factorConst{A}:= \inf_{H} \inf_{BC=A} \frac{\norm{X}{H}{C}\cdot \norm{H}{Y}{B}}{\norm{X}{Y}{A}}
\]
where the infimum runs over all Hilbert spaces $H$. We say $A$ factorizes through a Hilbert space if 
$\factorConst{A}< \infty$.
Further, let 
\[
    \factorSpConst{X}{Y} := \sup_{A} ~\factorConst{A}
\]
where the supremum runs over continuous operators $A:X\to Y$.  
As a quick example of the power of factorization theorems, observe 
that if $\id:X\to X$ is the identity operator on a Banach space $X$ and $\factorConst{\id}<\infty$, 
then $X$ is isomorphic to a Hilbert space and moreover the distortion (Banach-Mazur distance) is at most 
$\factorConst{\id}$ (\ie there exists an invertible operator $T:X \to H$ for some Hilbert space $H$ such 
that $\norm{X}{H}{T}\cdot \norm{H}{X}{T^{-1}}\leq \factorConst{\id}$). In fact (as observed by Maurey), 
\Kwapien gave an isomorphic characterization of Hilbert spaces by proving a factorization theorem. 

In this section we will show that our approximation results imply improved bounds on 
$\factorSpConst{\ell_{p}^{n}}{\ell_{q}^{m}}$ for certain values of $p$ and $q$. 
Before doing so, we first summarize prior work which will require the definitions of type and cotype: 
\begin{definition}
    The Type-2 constant of a Banach space $X$, denoted by $T_2(X)$, is the smallest constant $C$ such that 
    for every finite sequence of vectors $\{x^i\}$ in $X$, 
    \[
        \inparen{\Ex{\norm{\sum_{i} \eps_i\cdot x^i}^2}}^{1/2}  \leq C\cdot \sqrt{\sum_{i} \norm{x^i}^{2}}
    \]
    where $\eps_i$ is an independent Rademacher random variable. We say $X$ is of Type-2 if 
    $T_2(X)<\infty$. 
\end{definition} 
\begin{definition}
    The Cotype-2 constant of a Banach space $X$, denoted by $C_2(X)$, is the smallest constant $C$ 
    such that for every finite sequence of vectors $\{x^i\}$ in $X$, 
    \[
        \inparen{\Ex{\norm{\sum_{i} \eps_i\cdot x^i}^2}}^{1/2}  \geq \frac{1}{C}\cdot \sqrt{\sum_{i} \norm{x^i}^{2}}
    \]
    where $\eps_i$ is an independent Rademacher random variable. We say $X$ is of Cotype-2 if 
    $C_2(X)<\infty$. 
\end{definition} 
\begin{remark}~
    \begin{itemize}
        \item It is known that $C_2(X^*)\leq T_2(X)$. 
        
        \item It is known that for $p \geq 2$,  we have $T_2(\ell_p^n) = \gamma_p$ (while $C_2(\ell_p^n) \to
            \infty$ ~as~ $n \to \infty$) and for $q \leq 2$,~ $C_2(\ell_q^n) = \max\{2^{1/q-1/2},1/\gamma_q\}$ 
            (while~ $T_2(\ell_q^n) \to \infty$ ~as~ $n \to \infty$). 
    \end{itemize}
\end{remark}
\noindent
We say $X$ is Type-2 (resp. Cotype-2) if $T_2(X)<\infty$ (resp. $C_2(X)<\infty$). 
$T_2(X)$ and $C_2(X)$ can be regarded as measures of the ``closeness'' of $X$ to a 
Hilbert space. Some notable manifestations of this correspondence are: 
\begin{itemize}
    \item $T_2(X)=C_2(X)=1$ if and only if $X$ is isometric to a Hilbert space. 
    
    \item {\Kwapien~\cite{Kwapien72a}}:~~$X$ is of Type-2 and Cotype-2 if and only if it is isomorphic to 
    a Hilbert space. 
    
    \item {Figiel, Lindenstrauss and Milman~\cite{FLM77}}:~~If $X$ is a Banach space of Cotype-2, then any 
    $n$-dimensional subspace of $X$ has an $m=\Omega(n)$-dimensional subspace with Banach-Mazur 
    distance at most $2$ from $\ell_{2}^{m}$.  
\end{itemize}
Maurey observed that a more general factorization result underlies \Kwapien's 
work: 
\begin{theorem}[\Kwapien-Maurey]
    Let $X$ be a Banach space of Type-2 and $Y$ be a Banach space of Cotype-2. Then any 
    operator $T:X\to Y$ factorizes through a Hilbert space. Moreover $\factorSpConst{X}{Y}\leq 
    T_2(X)C_2(Y)$. 
\end{theorem}

Surprisingly Grothendieck's work which predates the work of \Kwapien and Maurey, established that 
$\factorSpConst{\ell_\infty^n}{\ell_1^m}\leq K_G$ for all $m,n\in \N$, which is not implied by the above 
theorem since $T_2(\ell_\infty^n) \to \infty$ as $n\to \infty$. Pisier~\cite{Pisier80} unified the above 
results for the case of approximable operators by proving the following: 
\begin{theorem}[Pisier]
    Let $X,Y$ be Banach spaces such that $X^*, Y$ are of Cotype-2. Then any approximable 
    operator $T:X\to Y$ factorizes through a Hilbert space. Moreover \\$\factorConst{T}\leq 
    (2\,C_2(X^*)C_2(Y))^{3/2}$. 
\end{theorem}

In the next section we show that for any $p^*,q\in [1,2]$, any $m,n\in \N$~  
\[
    \factorSpConst{\ell_{p}^{n}}{\ell_{q}^{m}} 
    \leq 
    \frac{1+\eps_0}{\sinh^{-1}(1)}\cdot C_2(\ell_{p^*}^{n})\cdot C_2(\ell_{q}^{m})
\] 
which improves upon Pisier's bound and for certain ranges of $(p,q)$, improves upon $K_G$ as well as the 
bound of \Kwapien-Maurey.

\subsection{Improved Factorization Bounds}
\label[section]{factorization:improvement}
In this section we will show that our approximation results imply improved bounds on 
$\factorSpConst{\ell_{p}^{n}}{\ell_{q}^{m}}$ for certain values of $p$ and $q$. 
To do so, we first require some preliminaries. We will give an exposition here of how upper bounds 
on integrality gaps yield factorization bounds for a very general class of normed spaces. 

\subsubsection{$p$-convexity and $q$-concavity}
\label{pconvexity}
The notions of $p$-convexity and $q$-concavity are well defined for a wide class of normed spaces known as Banach lattices. 
In this document we only define these notions for finite dimensional norms that are $1$-unconditional in the elementary basis 
(\ie those norms $X$ for which flipping the sign of an entry of $x$ does not change the norm. We shall refer to such norms 
as \emph{sign-invariant norms}). Most of the statements we make 
in this context can be readily extended to the case of norms admitting some $1$-unconditional basis, but we choose to fix the 
elementary basis in the interest of clarity. With respect to the goals of this document, we believe most of the key insights are 
already manifest in the elementary basis case. 
\medskip

\begin{definition}[$p$-convexity/$q$-concavity]
    Let $X$ be a sign-invariant norm over $\R^{n}$. Then for $1\leq p \leq \infty$ the $p$-convexity constant of $X$, 
    denoted by $M^{(p)}(X)$, is the smallest constant $C$ such that for 
    every finite sequence of vectors $\{x^i\}$ in $X$, 
    \[
        \left\| \left[\sum_{i} |[x^i]|^{p}\right]^{1/p}\right\| \leq~C\cdot \inparen{\sum_{i} \norm{x^i}^{p}}^{1/p}
    \]
    $X$ is said to be $p$-convex if $M^{(p)}(X)<\infty$. 
    We will say $X$ is exactly $p$-convex if $M^{(p)}(X) = 1$. 
    
    \medskip 
    
    \noindent
    For $1\leq q \leq \infty$, the $q$-concavity constant of $X$, 
    denoted by $M_{(q)}(X)$, is the smallest constant $C$ such that for every finite sequence of vectors $\{x^i\}$ in $X$, 
    \[
        \left\| \left[\sum_{i} |[x^i]|^{q}\right]^{1/q}\right\| \geq~\frac{1}{C}\cdot \inparen{\sum_{i} \norm{x^i}^{q}}^{1/q}.
    \]
    $X$ is said to be $q$-concave if $M_{(q)}(X)<\infty$. \\
    We will say $X$ is exactly $q$-concave if $M_{(q)}(X) = 1$. 
\end{definition}
\noindent
Every sign-invariant norm is exactly $1$-convex and $\infty$-concave. 
\medskip

\noindent
For a sign-invariant norm $X$ over $\R^n$, and any $0<p<\infty$ let $X^{(p)}$ denote the function $\norm{|[x]|^{p}}_{X}^{1/p}$. 
$X^{(p)}$ is referred to as the $p$-convexification of $X$. 
It is easily verified that $M^{(p)}(X^{(p)}) = M^{(1)}(X)$ and further that  
$X^{(p)}$ is an exactly $p$-convex sign-invariant norm if and only if $X$ is a sign-invariant norm (and therefore exactly $1$-convex).

\subsubsection{Convex Relaxation for Operator Norm} 
In this section we will see that there is a natural convex relaxation for a wide class of operator norms. 
It is instructive to first consider the pertinent relaxation for Grothendieck's inequality. 
Recall the bilinear formulation of the problem wherein given an $m\times n$ matrix $A$, the goal is to 
maximize $y^TA\,x$ over $\norm{\infty}{y}, \norm{\infty}{x} \leq 1$. One then considers
the following semidefinite programming relaxation: 
\begin{align*}
    \mbox{maximize} \quad&~~\sum_{i,j} A_{i,j}\cdot \mysmalldot{u^i}{v^j} \quad  \text{s.t.} \\
\mbox{subject to}    \quad&~~\norm{2}{u^i} \leq 1, \norm{2}{v^j} \leq 1 & \forall i\in [m], j\in [n]\\
    &~~ u^i, v^j\in \R^{m+n} & \forall i\in [m], j\in [n]
\end{align*}
which is equivalent to 
\begin{align*}
    \textbf{maximize} \quad
    &\frac{1}{2}\cdot 
    \mydot{
    \left[
    \begin{array}{cc}
    0 & A \\
    A^T & 0
    \end{array}
    \right]   
    }
    {
    \left[
    \begin{array}{cc}
    \bbY & \bbW \\
    \bbW^T & \bbX
    \end{array}
    \right]   
    } 
    \quad 
    \text{s.t.} \\[0.5em]
    &
    \bbX_{i,i}\leq 1 ,\quad
    \bbY_{j,j}\leq 1 \\[0.5em]
    &
    \left[
    \begin{array}{cc}
    \bbY & \bbW \\
    \bbW^T & \bbX
    \end{array}
    \right]   
    \succeq 0 ,\quad 
    \bbY\in \Sym^{m\times m},~\bbX\in \Sym^{n\times n},~\bbW\in \R^{m\times n}
\end{align*}
where $\Sym^{m\times m}$ is the set of $m\times m$ symmetric positive semidefinite matrices in $\R^{m\times m}$. 
\bigskip 

\noindent
Nesterov~\cite{Nesterov98, NWY00}\footnote{
Nesterov uses the language of quadratic programming and appears not to have noticed the connections 
to Banach space theory. In fact, it appears that Nesterov even gave yet another proof of an $O(1)$ upper bound 
on Grothendieck's constant. 
} and independently Naor and Schechtman\footnote{personal communication} observed that 
if $X$ and $Y^*$ are exactly $2$-convex, then there is a natural computable convex relaxation for the bilinear 
formulation of $X\to Y$ operator norm. Recall the goal is to maximize $y^TA\,x$ over 
$\norm{Y^*}{y}, \norm{X}{x} \leq 1$.
The relaxation which we will call $\CP{A}$  is as follows: 
\begin{align*}
\label{2convex:relaxation}
    \textbf{maximize} \quad
    &\frac{1}{2}\cdot 
    \mydot{
    \left[
    \begin{array}{cc}
    0 & A \\
    A^T & 0
    \end{array}
    \right]   
    }
    {
    \left[
    \begin{array}{cc}
    \bbY & \bbW \\
    \bbW^T & \bbX
    \end{array}
    \right]   
    } 
    \quad 
    \text{s.t.} \\[0.5em]
    &
    \mathrm{diag}(\bbX)\in \Ball{X^{(1/2)}} ,\quad
    \mathrm{diag}(\bbY)\in \Ball{{Y^*}^{(1/2)}} \\[0.5em]
    &
    \left[
    \begin{array}{cc}
    \bbY & \bbW \\
    \bbW^T & \bbX
    \end{array}
    \right]   
    \succeq 0 ,\quad 
    \bbY\in \Sym^{m\times m},~\bbX\in \Sym^{n\times n},~\bbW\in \R^{m\times n}
\end{align*}
For a vector $s$, let $D_s$ denote the diagonal matrix with $s$ as diagonal entries. 
Let $\overline{X} \defeq (X^{(1/2)})^*,~\overline{Y} \defeq ({Y^*}^{(1/2)})^*$. 
We can then define the dual program $\CPD{A}$ as follows: 
\begin{align*}
    & \textbf{minimize} \quad 
    (\norm{\overline{Y}}{s}+\norm{\overline{X}}{t})/2 \quad \text{s.t.}\\
    &\left[
    \begin{array}{cc}
    D_s & -A \\
    -A^T & D_t
    \end{array}
    \right]   
    \succeq 0 ,\quad 
    s\in \R^{m},~t\in \R^{n}\mper
\end{align*}
Strong duality is satisfied, i.e. $\CPD{A}=\CP{A}$, and a proof can be found in 
\cite{NWY00} (see Lemma 13.2.2 and Theorem 13.2.3).

\subsubsection{Integrality Gap Implies Factorization Upper Bound}
Known upper bounds on $\factorSpConst{X}{Y}$ involve Hahn-Banach separation arguments. 
In this section we see that for a special class of Banach spaces admitting a convex programming relaxation, 
$\factorSpConst{X}{Y}$ is bounded by the integrality gap of the relaxation as an immediate consequence 
of Convex programming duality (which of course uses a separation argument under the hood). 
A very similar observation has already been made by Tropp~\cite{Tropp09} in the special case of 
$X=\ell_{\infty}^{n}, Y=\ell_{1}^{m}$ with a slightly different convex program. 

For norms $X$ over $\R^n$, $Y$ over $\R^m$ and an operator $A:X\to Y$, we define 
\[
    \threeFactorConst{A} := \inf_{D_1BD_2 = A} 
    \frac{\norm{X}{2}{D_2}\cdot\norm{2}{2}{B}\cdot\norm{2}{Y}{D_1}}{\norm{X}{Y}{A}} \qquad 
    \threeFactorSpConst{X}{Y} := \sup_{A:X\to Y} \threeFactorConst{A}
\]
where the infimum runs over diagonal matrices $D_1,D_2$ and $B\in\R^{m\times n}$. 
Clearly, $\factorConst{A}\leq \threeFactorConst{A}$ and therefore $\factorSpConst{X}{Y} \leq 
\threeFactorSpConst{X}{Y}$. 

Henceforth let $X$ be exactly an exactly $2$-convex norm over $\R^n$ and $Y^*$ be an exactly 
$2$-convex norm over $\R^{m}$ (\ie $Y$ is exactly $2$-concave).
As was the approach of Grothendieck, we give an upper bound on 
$\factorSpConst{X}{Y}$ by giving an upper bound on 
$\threeFactorSpConst{X}{Y}$, which we do by showing 
\begin{lemma}
    Let $X$ be an exactly $2$-convex (sign-invariant) norm over $\R^n$ and $Y^*$ be an exactly 
    $2$-convex (sign-invariant) norm over $\R^{m}$. Then for any $A:X\to Y$, ~~
    $\threeFactorConst{A}~\leq ~\CPD{A}/\norm{X}{Y}{A}$.
\end{lemma}

\begin{proof}
    Consider an optimal solution to $\CPD{A}$. We will show 
    \[
        \inf_{D_1BD_2 = A} \norm{X}{2}{D_2}\cdot\norm{2}{2}{B}\cdot\norm{2}{Y}{D_1}
        \leq ~
        \CPD{A}
    \]
    by taking $D_1 := D_{s}^{1/2}$, ~$D_2 := D_{t}^{1/2}$ and 
    $B := \inparen{D_{s}^{1/2}}^{\dagger} A \inparen{D_{t}^{1/2}}^{\dagger}$ 
    (where for a diagonal matrix $D$, $D^{\dagger}$ only inverts the non-zero diagonal entries and 
    zero-entries remain the same). 
    Note that $s_i = 0$ (resp. $t_i=0$) implies the $i$-th row (resp. $i$-th column) of $A$ is all zeroes, since 
    otherwise one can find a $2\times 2$ principal submatrix (of the block matrix in the relaxation) 
    that is not PSD. This implies that $D_1 B D_2 = A$.  
    
    It remains to show that 
    $\norm{X}{2}{D_2}\cdot\norm{2}{2}{B}\cdot\norm{2}{Y}{D_1}\leq ~\CPD{A}$. 
    Now we have, 
    \[
        \norm{X}{2}{D_{t}^{1/2}}
        =
        \sup_{x\in \Ball{X}} \norm{2}{D_{t}^{1/2} x} 
        = 
        \sup_{x\in \Ball{X}} \sqrt{\mysmalldot{t}{[x]^{2}}}  
        =  
        \sup_{\widetilde{x}\in \Ball{X^{(1/2)}}} \sqrt{|\mysmalldot{t}{\widetilde{x}}|}  
        =
        \sqrt{\norm{X^{(1/2)}}{t}} \mper
    \]
    Similarly, since~ $\norm{2}{Y}{D_1} = \norm{Y^*}{2}{D_1}$~ we have 
    \[
        \norm{Y^*}{2}{D_{1}} 
        \leq  
        \sqrt{\norm{{Y^*}^{(1/2)}}{s}} \mper
    \]
    Thus it suffices to show $\norm{2}{2}{B}\leq 1$ since 
    \[
        \norm{X}{2}{D_2}\cdot\norm{2}{Y}{D_1}
        \leq 
        \sqrt{\norm{{X}^{(1/2)}}{t}\cdot \norm{{Y^*}^{(1/2)}}{s}} 
        \leq 
        (\norm{{Y^*}^{(1/2)}}{s}+\norm{{X}^{(1/2)}}{t})/2
        =
        \CPD{A} \mper
    \]
    We have, 
    \begin{align*}
        &\left[
        \begin{array}{cc}
        D_s & -A \\
        -A^T & D_t
        \end{array}
        \right]   
        \succeq 
        0 \\
        \Rightarrow
        &\left[
        \begin{array}{cc}
        \inparen{D_{s}^{1/2}}^{\dagger} & 0 \\
        0 & \inparen{D_{t}^{1/2}}^{\dagger}
        \end{array}
        \right]
        \left[
        \begin{array}{cc}
        D_s & -A \\
        -A^T & D_t
        \end{array}
        \right]
        \left[
        \begin{array}{cc}
        \inparen{D_{s}^{1/2}}^{\dagger} & 0 \\
        0 & \inparen{D_{t}^{1/2}}^{\dagger}
        \end{array}
        \right]   
        \succeq 
        0 \\
        \Rightarrow 
        &\left[
        \begin{array}{cc}
        D_{\bars} & -B \\
        -B^T & D_{\bart}
        \end{array}
        \right]
        \succeq 
        0 \qquad
        \text{for some ~} \bars\in \{0,1\}^{m},~ \bart \in \{0,1\}^{n} \\
        \Rightarrow 
        &\left[
        \begin{array}{cc}
        \id & -B \\
        -B^T & \id
        \end{array}
        \right]
        \succeq 
        0 \\
        \Rightarrow 
        &~\,\norm{2}{2}{B}
        \leq 
        1 \qedforce
    \end{align*}
\let\qed\relax
\end{proof}

\subsubsection{Improved Factorization Bounds for Certain $\ell_{p}^{n},\ell_{q}^{m}$}
Let $1\leq q \leq 2 \leq p\leq \infty$. Then taking $\sqnormX$ to be the $\ell_{p/2}^{n}$ unit ball and 
$\sqnormY$ to be the $\ell_{q^*/2}^{m}$ unit ball, we have $\sqrtnormX$ and $\sqrtnormY$ are 
respectively the unit balls in $\ell_{p}^{n}$ and $\ell_{q^*}^{m}$. Therefore $X$ and $Y$ as defined above 
are the spaces $\ell_{p}^{n}$ and $\ell_{q}^{m}$ respectively. Hence we obtain 
\begin{theorem}[$\ell_{p}^{n}\to \ell_{q}^{m}$ factorization]
\label[theorem]{p:q:factorization}
    If $1\leq q \leq 2\leq p\leq \infty$, then for any $m,n\in \N$ and $\eps_0= 0.00863$, 
    \[
        \factorSpConst{\ell_{p}^{n}}{\ell_{q}^{m}}
        ~\leq ~
        \frac{1+\eps_0}{\sinh^{-1}(1)\cdot \gamma_{p^*}\,\gamma_{q}}
        ~\leq ~
        \frac{1+\eps_0}{\sinh^{-1}(1)}\cdot C_2(\ell_{p^*}^{n})\cdot C_2(\ell_{q}^{m}).
    \]
\end{theorem}
\noindent
This improves upon Pisier's bound and for a certain range of $(p,q)$, improves upon $K_G$ as well as the 
bound of \Kwapien-Maurey. 
\bigskip

Krivine and independently Nesterov\cite{Nesterov98} observed that the integrality gap of $\CP{A}$ for any 
pair of convex sets $\sqnormX,\sqnormY$ is bounded by $K_G$ (Grothendieck's constant). This provides 
a class of Banach space pairs for which $K_G$ is an upper bound on the factorization constant. 
We include a proof for completeness. 

\subsubsection[Grothendieck Bound on Approximation Ratio]{$K_G$ Bound on Integrality Gap}
\label[subsection]{KG:worst:case}
In this subsection, we prove the observation that for exactly $2$-convex $X,Y^*$, 
the integrality gap for $X \to Y$ operator norm is always bounded by $K_G$. 

\begin{lemma}
    Let $X$ be an exactly $2$-convex (sign-invariant) norm over $\R^n$ and $Y^*$ be an exactly 
    $2$-convex (sign-invariant) norm over $\R^{m}$. Then for any $A:X\to Y$, ~~
    $\CP{A}/\norm{X}{Y}{A} ~\leq ~ K_G$.
\end{lemma}
\begin{proof}
Let 
$
B :=  \frac{1}{2}  \left[
    \begin{array}{cc}
    0 & A \\
    A^T & 0
    \end{array}
    \right]$.
The main intuition of the proof is to decompose $x \in X$ as $x = |[x]| \circ \sgn{[x]}$ 
(where $\circ$ denotes Hadamard/entry-wise multiplication), 
and then use Grothendieck's inequality on $\sgn{[x]}$ and $\sgn{[y]}$. Another simple observation is that 
for any convex set $\calF$, the feasible set we optimize over is invariant under {\em factoring out the 
magnitudes of the diagonal entries}. In other words, 
\begin{align}
    &\{ \Diag{d} \ \Sigma \ \Diag{d} : d \in \sqrt{|[\calF]|} \cap \R^n_{\geq 0},~ 
    \Sigma \succeq 0,~ \diag(\Sigma) = \1 \} \nonumber \\
    = 
    &\{ \bbX : \diag(\bbX) \in \calF,~ \bbX \succeq 0 \} \label[equation]{eq:feasible_sets}
\end{align}
We will apply the above fact for $\calF = \Ball{X^{(1/2)}} \oplus \Ball{{Y^*}^{(1/2)}}$. 
Let $X^+$ denote $X\cap \R_{\geq 0}^n$ (analogous for $(Y^*)^+$). 
Now simple algebraic manipulations yield 
\begin{align*}
    &\quad \norm{X}{Y}{A} \\
    =&\quad 
    \sup_{x \in X,~ y \in Y^*} (y \oplus x)^T B (y \oplus x) \\
    =&\quad 
    \sup_{\substack{d_x \in X^+,\  \sigma_x \in \{ \pm 1 \}^{n}, \\
    d_y \in (Y^*)^+, \ \sigma_y \in \{ \pm 1 \}^{m}}} 
    ((d_y \circ \sigma_y) \oplus (d_x \circ \sigma_x))^T ~B~ ((d_y \circ \sigma_y) \oplus (d_x \circ \sigma_x)) \\
    =&\quad 
    \sup_{\substack{d_x \in X^+, \ \sigma_x \in \{ \pm 1 \}^{n}, \\
    d_y \in (Y^*)^+, \ \sigma_y \in \{ \pm 1 \}^{m}}} 
    (\sigma_y \oplus \sigma_x)^T  (\Diag{d_y \oplus d_x} ~B~ \Diag{d_y \oplus d_x}) 
    (\sigma_y \oplus \sigma_x) \\
    \geq &\quad 
    (1/K_G) \cdot 
    \sup_{\substack{d_x \in X^+, \ d_y \in (Y^*)^+, \\
    \Sigma :~ \diag(\Sigma) = \1, ~\Sigma \succeq 0}}  
    \mydot{ \Sigma~}{~\Diag{d_y \oplus d_x} ~B~  \Diag{d_y \oplus d_x}} 
    &&(\text{Grothendieck})\\
    = &\quad 
    (1/K_G) \cdot 
    \sup_{\substack{d_x \in X^+, \ d_y \in (Y^*)^+, \\
    \Sigma :~ \diag(\Sigma) = \1, ~\Sigma \succeq 0}}  
    \mydot{\Diag{d_y \oplus d_x} ~\Sigma ~  \Diag{d_y \oplus d_x}~}{ ~B} \\
    =&\quad 
    (1/K_G) \cdot \CP{A} &&(\text{by \cref{eq:feasible_sets}}) 
    \qedhere
\end{align*}
\end{proof}

\bibliographystyle{alpha}
\bibliography{p-to-q-refs}

\end{document}